\documentclass[10pt,letterpaper]{article}
\usepackage[top=0.85in,left=2.75in,footskip=0.75in]{geometry}

\usepackage{amsmath,amssymb,amsthm}

\usepackage{changepage}

\usepackage{textcomp,marvosym}

\usepackage{cite}

\usepackage{nameref,hyperref}

\usepackage[right]{lineno}

\usepackage[nopatch=eqnum]{microtype}
\DisableLigatures[f]{encoding = *, family = * }

\usepackage[table]{xcolor}

\usepackage{array}

\usepackage[inline]{enumitem}
\usepackage{thm-restate}
\usepackage{subfigure}
\usepackage{makecell}

\newtheorem{Corollary}{Corollary}
\newtheorem{Example}{Example}
\newtheorem{Lemma}{Lemma}

\newcolumntype{+}{!{\vrule width 2pt}}

\newlength\savedwidth

\newcommand\thickhline{\noalign{\global\savedwidth\arrayrulewidth\global\arrayrulewidth 2pt}%
\hline
\noalign{\global\arrayrulewidth\savedwidth}}

\newcommand{\pospart}[1]{\left\vert\,{#1}\,\right\vert_+}
\newcommand{\absval}[1]{\left\vert\,{#1}\,\right\vert}
\usepackage{setspace} 
\raggedright
\usepackage[aboveskip=1pt,labelfont=bf,labelsep=period,justification=raggedright,singlelinecheck=off]{caption}

\makeatletter
\renewcommand{\@biblabel}[1]{\quad#1.}
\makeatother

\usepackage{fancyhdr,graphicx}
\usepackage{epstopdf}
\fancyheadoffset[L]{2.25in}
\fancyfootoffset[L]{2.25in}
\begin{document}
\vspace*{0.2in}

\begin{flushleft}
{\Large
\textbf\newline{Upper bounds for integrated information} 
}
\newline
\\
Alireza Zaeemzadeh\textsuperscript{1*}
 and
Giulio Tononi\textsuperscript{1*}
\\
\bigskip
\textbf{1} Department of Psychiatry, University of Wisconsin, Madison, WI 53719, USA
\\
\bigskip

%
%





* zaeemzadeh@wisc.edu (AZ), gtononi@wisc.edu (GT)

\end{flushleft}
\section*{Abstract}
Originally developed as a theory of consciousness, integrated information theory provides a mathematical framework to quantify the causal irreducibility of systems and subsets of units in the system. 
Specifically, \emph{mechanism integrated information} quantifies how much of the causal powers of a subset of units in a state, also referred to as a mechanism, cannot be accounted for by its parts. 
If the causal powers of the mechanism can be fully explained by its parts, it is reducible and its integrated information is zero. 
Here, we study the upper bound of this measure and how it is achieved. 
We study mechanisms in isolation, groups of mechanisms, and groups of causal relations among mechanisms. 
We put forward new theoretical results that show  mechanisms that share parts with each other cannot all achieve their maximum. 
We also introduce techniques to design systems that can maximize the integrated information of a subset of their mechanisms or relations. Our results can potentially be used to exploit the symmetries and constraints to reduce the computations significantly and to compare different connectivity profiles in terms of their maximal achievable integrated information.

\section*{Author summary}
Integrated Information Theory (IIT) offers a theoretical framework to quantify the causal irreducibilty of a system, subsets of the units in a system, and the causal relations among the subsets.
For example, \emph{mechanism integrated information} quantifies how much of the causal powers of a subset of units in a state cannot be accounted for by its parts. Here, we provide theoretical results on the upper bounds for this measure, how it is achieved, and why mechanisms with overlapping parts cannot all be maximally integrated. We also study the upper bounds for integrated information of causal relations among the mechanisms. The ideas introduced here can potentially pave the way to design systems with optimal causal irreducibility and to develop computationally lightweight exact or approximate measures for integrated information.


\section{Introduction}
\label{sec:intro}
Integrated information theory (IIT) has been developed as a comprehensive theory of what it takes for a system to be conscious, how much, and in which way \cite{Albantakis2023IntegratedTerms, Tononi2016IntegratedSubstrate}. 
The theory starts from the existence of experience and characterizes its essential properties—those that are true of every conceivable experience—called phenomenal ‘axioms’. These are as follows:  every experience is intrinsic (for the subject), specific (this one), unitary (a whole, irreducible to its parts), definite (this whole, having a border and grain), and structured (being composed of phenomenal distinctions and relations) \cite{Albantakis2023IntegratedTerms}. 
The theory then formulates the axioms in operational terms as ‘postulates’ of cause-effect power. Given a system’s causal model (a set of units in its current state, together with its transition probability matrix (TPM)), IIT’s postulates can then be employed to identify a substrate of consciousness or ‘complex’—a maximally irreducible set of units (having maximum \emph{system integrated information} $\varphi_s$) with its specific intrinsic cause-effect state. 
Finally, IIT’s postulates are employed to ‘unfold’ the cause-effect structure specified by the complex in its current state—the set of causal distinctions (cause-effects) specified by subsets of units within the complex, as well as their relations (the way their cause-effects overlap). 
According to IIT, the composition of the cause-effect structure corresponds to the quality of an experience—how the experience feels—and the sum of the integrated information values of its composing distinctions ($\varphi_d$) and relations ($\varphi_r$) corresponds to the quantity of consciousness ($\Phi$)—how much an entity exists intrinsically (for itself).

As shown in other work, the theory has explanatory, predictive, and inferential power \cite{Tononi2016IntegratedSubstrate}. For example, it explains why certain parts of the brain, but not others, can support consciousness, and why consciousness is lost during dreamless sleep and anesthesia \cite{Massimini2005BreakdownSleep, Pigorini2015BistabilitySleep}.
The theory has also been employed to account for the quality of experience, namely the pervasive feeling of spatial extendedness \cite{Haun2019WhyExperience} and the feeling of temporal flow \cite{Comolatti_Time}. 
It has led to clinical applications, using crude proxies of the system integrated information $\varphi_s$ that nevertheless offer what is currently the most sensitive and specific test for the presence of consciousness in non-responsive patients \cite{Casarotto2016StratificationComplexity, Massimini2005BreakdownSleep, Sarasso2020LocalInjury}. 
Finally, to the extent that the theory continues to be validated empirically in humans, it supports inferences about the presence, quantity, and quality of consciousness in other species as well as in artifacts. For example, it can be shown that computer architectures cannot be conscious in any meaningful sense because they break down into small complexes, each of which has a trivial value of $\Phi$, regardless of their ability to simulate intelligent behaviors and functions \cite{Findlay2019}.

IIT is unique in providing an exact calculus for the quantity and quality of consciousness—from first principles and based on phenomenology. However, unfolding cause-effect structures and determining the associated value of $\Phi$ exactly is not feasible for realistic systems, for several reasons. 
Among them is the difficulty of obtaining a system TPM at the right grain, the nested combinatorial explosions in assessing candidate unit grains and candidate complexes, as well as the composing distinctions and relations. For these reasons, and not unlike well-known precedents in statistical physics and quantum mechanics, it is essential to develop approximations and heuristics. 
These can then be used to estimate $\Phi$  and related quantities based on simple properties of various substrates, including the density and pattern of connections as well as various symmetries. As a step in this direction, it is important to begin establishing bounds on IIT’s basic quantities, which is the goal of this study. By obtaining such bounds, and progressively tightening them based on various properties of substrates of interest, we should ultimately be able to make well-grounded estimates about the presence and quantity of consciousness in different regions of the human brain, allowing us to more precisely test the theory’s predictions. 

A further goal will be to estimate the value of $\Phi$ in brains markedly different from ours, as well as in other natural and artificial systems. An important outcome of the search for bounds is the determination of orders of magnitudes for $\Phi$. As shown here, sums of integrated information of distinctions and relations can grow hyper-exponentially with the number of units. Therefore, a system with an architecture that allows a large number of units to constitute a maximally irreducible complex, should yield hyper-astronomical values of $\Phi$. 
We have conjectured that this should be the case for a densely connected lattice of units such as that found in posterior-central regions of the cerebral cortex \cite{Tononi2016IntegratedSubstrate, Albantakis2023IntegratedTerms}. 
In contrast, large systems with fault lines, or high levels of indeterminism and/or degeneracy, break down into many small complexes, which will necessarily have very low values of $\Phi$. 
This should be the case for many other regions of the brain, such as the cerebellum and much of prefrontal cortex, for other parts of the body, and certainly for artificial systems such as computers. The expected hyper-astronomical difference in $\Phi$ values between these substrates is essential for providing some principled guidelines about the occurrence of consciousness in nature and for informing the ongoing debate about panpsychism \cite{Tononi2015Consciousness:Everywhere}. 

\subsection{Outline}
The measures to quantify integrated information of distinctions $\varphi_d$ and relations $\varphi_r$ that capture the postulates of IIT are presented in \cite{Albantakis2023IntegratedTerms} and are discussed in detail in Section \ref{sec:mechanism_bounds} and Table \ref{tab:intro}. According to IIT, the causal components within the system must satisfy the same properties as the system except composition: they must have cause-effect power within the system (intrinsically), select a specific state (information) in a way that cannot be accounted for by their parts (integration) and over a definite set of units (exclusion) \cite{Albantakis2023IntegratedTerms}. Formalism of IIT enables us to quantify the causal \emph{irreducibility} of each subset of system. A set of units in its current state, also referred to as a \emph{mechanism}, is irreducible if its causal powers cannot be accounted for by its parts. For that, the potential cause and the potential effect of each mechanism are identified and the irreducibility is quantified by measuring how such cause or effect can be accounted for only using the parts of the mechanism. 
A central quantity in the formalism of IIT is the mechanism integrated information, denoted by $\varphi$, which measures the causal irreducibility of a single mechanism \cite{Albantakis2023IntegratedTerms,Barbosa2021MechanismInformation}. An irreducible mechanism, \emph{i.e.,} a mechanism with nonzero $\varphi$, with its corresponding cause and effect is referred to as a \emph{distinction}.

IIT also provides us with a framework to quantify how different distinctions causally interact with each other by defining and measuring the strength of the \emph{relations} among them \cite{Albantakis2023IntegratedTerms,Haun2019WhyExperience}. For example, distinctions might have overlapping effects, making them  related. 
Finally, the causal powers of the whole system can be accounted for by its cause-effect structure, which is composed of its causal distinctions and causal relations that bind together the distinctions.

In this work we study the upper bounds achievable by these measures and how we can achieve them.
Barbosa et al \cite{Barbosa2021MechanismInformation, Barbosa2020AInformation} showed the integrated information of an individual distinction can increase by adding reliable, not noisy, units to it. Here, we study the maximum achievable integrated information of a single distinction of fixed size and the trade-offs when considering the system as a whole, with all its overlapping distinctions and their shared parts. 
In particular, we discus the following bounds:
\begin{itemize}
    \item In Section \ref{subsec:single_mechanism}, we derive an upper bound for how much information a distinction can specify beyond each of its parts. 
    \item We use this bound to find the maximum integrated information achievable by a distinction, as well as a bound for the sum of integrated information of all the distinctions.
    \item Then in Section \ref{subsec:multiple_mechanisms} and Section \ref{subsec:selectivity_1}, we will show why the distinctions of a system cannot all achieve their corresponding maximum $\varphi$. We further provide a numerical bound for a special class of systems with grid-like connectivity pattern.
    \item The upper bounds for integrated information of relations, as well as the conditions necessary for achieving them, are presented in Section \ref{sec:relation_bounds}.
\end{itemize}

Section \ref{sec:experiments} provides numerical experiments and discussions on mathematical properties of the bounds and the constructions. Finally, implications of this study, a few open problems, and potential directions for research are discussed in Section \ref{sec:discussion}. 

\begin{table}[!ht]
\begin{adjustwidth}{-2.25in}{0in} 
\centering
\caption{
{\bf A summary of IIT concepts and the relevant notation used in this manuscript.}}
\small
\begin{tabular}{|m{3.5cm}|m{14cm}|}
\hline
Candidate system & A set of $N$ random variables $S = \{S_1, S_2, \dots, S_N \}$ with discrete updates and discrete state space $\Omega_S$. The system dynamics at time $t$ is defined by a transition probability matrix (TPM) $p(S_{t+1} = s_{t+1} | S_{t} = s_t)$. $s_t$ and $s_{t+1}$ denote the system state at time $t$ and $t+1$, respectively.\\ \hline
Mechanism & A subset of units in the system $M \subseteq S$ in a state $m$ with state space $\Omega_M$. The state of the mechanism is inherited from the system state $s \in \Omega_S$.\\ \hline
Purview & A subset of units $Z \subseteq S$ with state space $\Omega_Z$ whose states are constrained by a mechanism. A mechanism $M$ at time $t$ in state $m_t$ constrains the state of its potential cause purview $Z_{t-1} \subseteq S$ and its potential effect purview $Z_{t+1} \subseteq S$. The states of the purviews are denoted by $z_{t-1}$ and $z_{t+1}$. To avoid cluttering, time subscripts are often dropped.\\ \hline
Cause/effect repertoire & The effect (cause) repertoire $\pi_e(Z_{t+1} \mid M_t=m_t)$ ($\pi_c(Z_{t-1} \mid M_t=m_t)$) is defined as the probability distribution over the potential effect purview $Z_{t+1}$ (the potential cause purview $Z_{t-1}$), given that the mechanism is in state $m_t$, and is obtained by \emph{causally} marginalizing out the random variables outside the mechanism and purview. \\ \hline
Causal marginalization & The process of rendering a subset of units causally inert. The units are marginalized based on a uniform marginal distribution. The process is repeated separately for each purview unit, and they are then combined using a product,
which eliminates any residual correlations from the marginalized units with divergent connections. We formally define the process of causal marginalization for units outside a mechanism $W = S - M$ in Eq. \eqref{eq:single_unit_pi_e} and Eq. \eqref{eq:pi_e}. \\ \hline
Maximal cause/effect state & Given a mechanism $M$ in state $m$ and an effect (cause) purview $Z$, maximal effect (cause) state $z'_e(m, Z)$ ($z'_c(m, Z)$) is the state of $Z$ for which the mechanism $m$ makes the most difference compared to chance (see Eq. \eqref{eq:purview_state_e_def} \\ \hline 
Partition & A partition $\theta$ cuts a mechanism-purview pair into independent parts by causally marginalizing a subset of connections between them. $\Theta(M,Z)$ is the set of all the valid partitions for a given mechanism-purview pair and $\mathcal{N}(\theta)$ is the number of connections severed by partition $\theta$.\\ \hline
Integrated cause/effect information & The integrated effect (cause) information of $m$ over effect (cause) purview $Z$, denoted by $\varphi_e(m, Z)$ ($\varphi_c(m, Z)$) quantifies how much difference the mechanism makes to the maximal effect (cause) state above and beyond its parts, which requires a search over all possible partitions $\Theta(M, Z)$ (See Eq. \eqref{eq:purview_phi_e_def} and Eq. \eqref{eq:purview_phi_c_def}).\\ \hline
Maximally irreducible cause/effect & The maximally irreducible effect (cause) of mechanism $m$, denoted by $z_e^*(m)$ ($z_c^*(m)$) is the maximal state of effect (cause) purview $Z$ which has the maximum integrated effect (cause) information, compared to other purviews.\\ \hline
Mechanism integrated cause/effect information & Mechanism integrated effect (cause) information $\varphi_e(m)$ ($\varphi_c(m)$) is the integrated effect (cause) information of mechanism $m$ over its maximally irreducible effect (cause).\\ \hline
Mechanism integrated information & $\varphi(m) = \min \{ \varphi_e(m), \varphi_c(m) \}$. A mechanism is irreducible if it has nonzero mechanism integrated information, which means it has both an irreducible effect and an irreducible cause.\\ \hline
Causal distinction & An irreducible mechanism $m$ with its maximally irreducible cause and effect purviews defines a causal distinction $d(m) = (z^*(m), \varphi(m) )$, where $
z^*(m) = \left( z^*_c(m), z^*_e(m) \right)$ contains the maximally irreducible cause and effect purviews in their maximal states.\\ \hline
Causal relation & Any subset of distinctions in the system forms a relation if the cause purview, the effect purview, or both the cause and the effect purviews of each distinction in the subset overlap over the same units in the same states.\\ \hline
Cause-effect structure & The union of all causal distinctions in the system and the causal relations that bind together the distinctions.\\ \hline
\end{tabular}
\label{tab:intro}
\end{adjustwidth}
\end{table}
\section{Upper bounds for mechanism integrated information}
\label{sec:mechanism_bounds}
Consider a stochastic system $S$ consisting of $N$ random variables $\{S_1, S_2, \dots, S_N \}$. These random variables represent a system with transition probability matrix (TPM) defined as $p(S_{t+1} = s_{t+1} \mid S_{t} = s_{t})$, which denotes the probability that the system is in state $s_{t+1}$ at time $t+1$ given the state of the system at time $t$. The mechanism integrated information, presented in \cite{Albantakis2023IntegratedTerms}, quantifies how much a mechanism $M \subseteq S$ in state $m_t$ constrains the state of its potential causes $Z_{t-1} \subseteq S$ and its potential effects $Z_{t+1} \subseteq S$, above and beyond their parts. For that, a difference measure is developed that compares the probability distribution of a cause purview $Z_{t-1}$ or an effect purview $Z_{t+1}$, before and after partitioning the mechanism-purview pair into independent parts. 

Formally, the effect repertoire $\pi_e(Z_{t+1} \mid M_t=m_t)$ is defined as the probability distribution over the potential effect purview $Z_{t+1}$, given that the mechanism is in state $m_t$, and is obtained by \emph{causally} marginalizing out the random variables outside the mechanism and purview, under the assumption that the variables at $t+1$ are conditionally independent given all the variables at $t$  \cite{Albantakis2023IntegratedTerms, Barbosa2021MechanismInformation}. Having introduced the effect repertoire, the maximal effect state of $m_t$ within the purview $Z_{t+1}$ is defined as:
\begin{equation}
    \label{eq:purview_state_e_def}
    z'_e(m, Z) =\arg \max_{z \in \Omega_Z} \pi_e(z \mid m)\log_2\left(\frac{\pi_e(z \mid m)}{\pi_e(z;M)}\right) 
\end{equation}
where $\pi_e(z;M)$ is the unconstrained effect probability of state $z \in \Omega_{Z}$:
$$
\label{eq:unconstrained_effect}
    \pi_e(z;M) = |\Omega_{M}|^{-1} \sum_{m \in \Omega_{M}} \pi_e(z \mid M = m) , \quad z \in \Omega_{Z}.
$$
To avoid cluttering the notation, the time subscripts $t+1$ and $t$ are dropped here. Eq \eqref{eq:purview_state_e_def} gives us the state for which the mechanism $m$ makes the most difference compared to chance. Since there is at least one state $z \in \Omega_{Z}$ such that $\pi_e(z \mid m) > \pi_e(z;M)$, the maximal effect state will always be a state for which the mechanism increases the probability compared to the unconstrained probability. This is in line with our intuition that for a state to be caused by a mechanism, its probability needs to be increased by that mechanism. Given this maximal effect state, the integrated effect information of $m$ quantifies how much difference the mechanism makes to this state above and beyond its parts, by comparing the effect repertoire for state $z'_e$, $\pi_e(z'_e \mid m)$, with its partitioned version $\pi_e^{\theta'}(z'_e \mid m)$:
\begin{equation}
    \label{eq:purview_phi_e_def}
    \varphi_e(m,Z) = \pi_e(z'_e \mid m)
    \pospart{\log_2\left(\frac{\pi_e(z'_e \mid m)}{\pi_e^{\theta'}(z'_e \mid m)}\right)}.
\end{equation}
Here, $|.|_+$ represents the positive part operator, which sets the negative values to $0$. This reflects the intuition that the mechanism as a whole needs to increase the probability of its effect, compared to the mechanism's parts \footnote{In Appendix \ref{app:other_measures}, we show that our results hold, even if we replace the positive part operator with the absolute value operator.}. $\pi_e^{\theta}(Z\mid m)$ is the effect repertoire calculated after partitioning the mechanism-purview pair into independent parts using the partition $\theta \in \Theta(M,Z)$, $\Theta(M,Z)$ is the set of all the valid partitions, and $\theta'$ is the minimum information partition (MIP). $\varphi_e(m,Z)$ quantifies how much the mechanism as a whole increases the probability of the effect state $z'_e$ compared to the MIP, which is obtained as:
\begin{equation}
\label{eq:mip}
\theta' = \arg \min_{\theta \in \Theta(M,Z)}
\frac{1}{\mathcal{N}(\theta)}
\pi_e(z'_e \mid m)
\pospart{\log_2\left(\frac{\pi_e(z'_e \mid m)}{\pi_e^{\theta}(z'_e \mid m)}\right)}.
\end{equation}
$\mathcal{N}(\theta)$ is the maximum possible distance between $\pi_e(z'_e \mid m)$ and $\pi_e^{\theta}(z'_e \mid m)$ achievable by partition $\theta$ and its value will be derived shortly in Lemma \ref{lem:num_connections}. Normalizing the distance between $\pi_e(z'_e \mid m)$ and $\pi_e^{\theta}(z'_e \mid m)$ by $\mathcal{N}(\theta)$ makes the comparison between different partitions with different number of parts fair. This is because partitions that severe fewer causal connections tend to change the probability less. The MIP is the partition $\theta$ that makes the least difference to the effect repertoire Normalized by $\mathcal{N}(\theta)$. 

If there exists a partition for which the unpartitioned probability is less than or equal to the partitioned probability, the mechanism is reducible to its parts and the integrated effect information over purview $Z$ is $0$.
We can further find the most irreducible purview by finding 
$$
z_e^*(m) = \arg\max_{\{z'_e \mid Z \subseteq S\}} \varphi_e(m, Z = z'_e).
$$
$z_e^*(m)$ is the subset of units that mechanism $m$ as a whole makes the most difference to.
$\varphi_e(m) = \varphi_e(m, z_e^*(m))$ is the difference that $m$ as a whole makes to its most irreducible purview, which is referred to as the mechanism integrated effect information. Similar analysis and procedure can be used to define the cause repertoire $\pi_c(Z_{t-1} \mid M_t=m_t)$ and the integrated cause information $\varphi_c(m)$ \cite{Albantakis2023IntegratedTerms}. For a detailed description of the definitions for the cause side, see \nameref{S1_Appendix}. Finally, the overall integrated information of a mechanism is then defined as $\varphi(m) = \min \{ \varphi_e(m), \varphi_c(m) \}$. An irreducible mechanism with its maximal cause and effect purviews are referred to as a \emph{distinction}.

It is evident that the difference measure of the form
$p(z\mid m) 
\pospart{\log_2 ( \frac{p(z\mid m)}{q(z \mid m)})}$
plays a central role in measuring $\varphi$. This was derived from the postulates of IIT and was first introduced in \cite{Barbosa2020AInformation}. In short, this measure satisfies the following properties:
\begin{enumerate*}[label=(\roman*)]
    \item The measure differs from $0$ only if the probability of the state is increased.
    \item The measure is not an aggregate over all the states and reflects how much change is made in an individual state.
    \item In a scenario where $p(z_1 \cup z_2 \mid m) = p_1(z_1 \mid m )p_2(z_2 \mid m)$, $q(z_1 \cup z_2 \mid m) = q_1(z_1 \mid m )q_2(z_2 \mid m)$, and $p_2(z_2 \mid m ) = q_2(z_2 \mid m)$, the measure produces a smaller value for the purview $z_1 \cup z_2$ than only $z_1$. This scenario represents the case where $m$ makes no difference to a subset of the purview $z_2$, therefore having this subset in the maximal purview is discouraged.
\end{enumerate*}

We can look at this measure as the product of two terms. The first term, $p(z \mid m)$, is referred to as selectivity, while the second term,
$
\pospart{
\log_2 ( \frac{p(z \mid m)}{q(z \mid m)})
}
$,
is called informativeness. Adding new units can never increase the selectivity, therefore the measure is only increased if the new units increase the informativeness enough. 
Variations of this measure have been used to define integrated information of distinctions and systems \cite{Albantakis2023IntegratedTerms,Marshall2023SystemInformation}. However, the question of how large these measures can get remains open. In this section, we first study the maximum integrated information achievable by a single distinctions. Then we show why this upper bound cannot be achieved by all the distinctions of a system, by  studying a few important special cases. 

Our working assumption to derive these bounds is that the system is realizable by a TPM that is a product of unit TPMs (conditional independence). This is a minimal assumption as both the definition of $\varphi$ and causal marginalization process make use of such TPM \cite{Barbosa2021MechanismInformation, Albantakis2023IntegratedTerms, Marshall2023SystemInformation}. The conditional independence reflect the assumption that the state of the units only depend on the previous time step and there is no instantaneous causation. Furthermore, we consider systems consisting of binary units, but the results are generalizable to non-binary units as well.
In \nameref{S2_Appendix}, we show that our results still hold even if we use slightly different difference measures such as point-wise mutual information or Kullback–Leibler divergence.

\subsection{Single mechanism}
\label{subsec:single_mechanism}

Before discussing our main results, we first need to present a few helpful lemmas. Lemma \ref{lem:pi_ratio} states how the process of causal marginalization places a certain limit on the informativeness. Formally, given the set of units outside the mechanism $W = S-M$, the effect repertoire of a single unit $Z_i \in Z$ is defined as:
\begin{equation}
\label{eq:single_unit_pi_e}
\pi_e(Z_i \mid m) = \frac{1}{2^{|W|}} \sum_w p(Z_i \mid m, w),
\end{equation}
where $|W| = N - |M|$ is the number of units outside the mechanism. Furthermore, the effect repertoire of $Z$ is defined as:
\begin{equation}
\label{eq:pi_e}
\pi_e(Z \mid m ) = \prod_i \pi_e(Z_i \mid m).
\end{equation}
Using this definition, it can be shown that:

\begin{restatable}[]{Lemma}{piratiolem}
\label{lem:pi_ratio}
Given two mechanisms $\bar{M}$ and $M$, such that $\bar{M} \subset M \subseteq S$, and a single unit $Z_i$, we have: 
$$\pospart{ \log_2(\frac{\pi_e(Z_i=z_i \mid m)}{\pi_e(Z_i=z_i \mid \bar{m})}) } \leq |M| - |\bar{M}|.$$
\end{restatable}

All the proofs are provided in \nameref{S3_Appendix}. 
$|M| - |\bar{M}|$ is the number of units that need to be causally marginalized to calculate $\pi_e(Z_i=z_i \mid \bar{m})$ from $\pi_e(Z_i=z_i \mid m)$, which can also be thought of as the number of causal connections cut from the unit $Z_i$. 
Throughout this manuscript, cutting a (causal) connection between a unit at $t$ and a unit at $t+1$ refers to recalculating the conditional probability distribution of the output by causally marginalizing out the input unit.
Furthermore, this bound does not depend on the state of the purview unit $z_i$ and it holds for any state, not just the maximal state selected in Eq \eqref{eq:purview_state_e_def}.  
The result in Lemma \ref{lem:pi_ratio} can be generalized as:
\begin{restatable}[]{Lemma}{numconnectionslem}
\label{lem:num_connections}
Given a mechanism $M$ in state $m$, a purview $Z$ in state $z$, and a partition $\theta$, we have: 
$$\pospart{ \log_2(\frac{\pi_e(Z=z \mid m)}{\pi_e^\theta(Z=z \mid m)}) } \leq \mathcal{N}(\theta),$$
where $\mathcal{N}(\theta)$ is the total number of connections cut by the partition $\theta$.
\end{restatable}
This result is general in the sense that it holds for any partitioning that removes an arbitrary subset of connections, even if it is not a valid partition and does not divide the mechanism-purview pair into independent parts. Therefore, it does not depend on the constraints imposed on the partitions. Lemma \ref{lem:pi_ratio} and Lemma \ref{lem:num_connections} establish a connection between the number of connections severed by a partition and the value of the integrated information. This can help us to develop a more intuitive understanding of mechanism integrated information. In words, the mechanism integrated information, as defined in \cite{Albantakis2023IntegratedTerms}, counts the number of causal connections that need to be severed to disintegrate the mechanism into causally independent parts.

As shown in Eq \eqref{eq:mip}, to find the partition that makes the least difference, we normalize the difference between the unpartitioned and partitioned repertoires by $\mathcal{N}(\theta)$. This makes the comparison between the partitions that sever different number of connections fair.
Using Lemma \ref{lem:num_connections} and the fact that the informativeness term is a bound for the integrated information, we can readily state our first main result. Theorem \ref{thm:single_mechanism} states that the integrated information of mechanism $M$ over a purview cannot be larger than the total number of potential connections between them. In other words, to disintegrate a maximal integrated mechanism we need to severe \emph{all} causal connections between the mechanism and the purview. This holds both for the cause and the effect sides. 

\begin{restatable}{Theorem}{singlemechthm}
\label{thm:single_mechanism}
For a mechanism $M \subseteq S$ in state $m$, a candidate cause purview $C$, and a candidate effect purview $E$, we have:
$$
\begin{aligned}
\varphi_e(m, E) \leq |M||E| \quad \text{and} \quad
\varphi_c(m, C) \leq |M||C|,
\end{aligned}
$$
where $|E|$ and $|C|$ denote the size of the candidate effect and cause purviews, respectively. 
\end{restatable}

This bound is achievable and the conditions to achieve it are presented in the proof. Theorem \ref{thm:single_mechanism} provides us with our first upper bound for the sum of integrated information of all the distinctions of a system:
\begin{equation}
\label{eq:sum_bound_K_nonunique}
\sum_{M \subseteq S} \varphi(m)
= \sum_{M \subseteq S} \min \{ \varphi_c(m), \varphi_e(m) \} 
\leq \sum_{M \subseteq S} |M| N
= N \sum_{|M| = 1}^N |M| \binom{N}{|M|}
 = \frac{N^2}{2} 2^{N}.
\end{equation}
And, if we are interested in the upper bound over unique purviews:
\begin{equation}
\label{eq:sum_bound_K}
\sum_{M \subseteq S} \varphi(m)
\leq \sum_{M \subseteq S} |M||Z^*(m)|
\overset{(a)}{\leq} \sum_{M \subseteq S} |M|^2
= \sum_{|M| = 1}^N |M|^2 \binom{N}{|M|}
 = \frac{N(N+1)}{4} 2^{N}.
\end{equation}
Inequality (a) follows from the fact that for any two sets $S_1$ and $S_2$, $|S_1|^2 + |S_2|^2 \geq |S_1||S_2| + |S_2||S_1|$. This means that in a scenario where each purview can be assigned to only one mechanism, matching the sizes of the mechanisms and purviews maximizes the upper bound. This scenario is an important special case, as it was studied in \cite{Haun2019WhyExperience} to explain how IIT can be used to account for the quality of certain type of experiences, such as the spatial experience. Both of these bounds consist of a quadratic term and an exponential term. This is because the total number of connections in a system grows quadratically and the number of subsets/mechanisms grows exponentially with the number of units in the system. This also emphasizes the fact that the notion of integrated information is fundamentally different from the Shannon information encoded in the state of the system. Shannon information of a system consisting of $N$ binary variables can at most be $N$. In the next section, we will show due to the constraints imposed by overlapping distinctions of different sizes, also referred to as mechanisms of different orders, these bounds are not achievable.

It is also worthwhile to mention that the results provided here for a single mechanism are in line with the bounds for the system integrated information defined in \cite{Marshall2023SystemInformation}. As shown in Theorem 1 in \cite{Marshall2023SystemInformation}, the maximum achievable system integrated information for a given partition is the number of connections cut by that partition. This makes the overall maximum system integrated information the maximum number of connections cut by any valid cut, which is $|S|(|S|-1)$ for the system partitions considered in \cite{Marshall2023SystemInformation}. 

\subsection{Inter-order constraints}
\label{subsec:multiple_mechanisms}
In this section, we discuss how mechanisms that are either subsets or supersets of each other place certain constraints on the maximum integrated information achievable by them. We refer to this set of constraints as inter-order constraints, as they are resulted from interactions among mechanisms of different order/size. The results presented in this section formalize another property of integrated information: the tension between the parts and the whole. We will show that either the whole or the parts can be maximally integrated. As already shown in Lemma \ref{lem:pi_ratio}, the levels of determinism for a single purview unit given two different mechanisms are bounded by each other, if one mechanism is a subset of the other. It can also be shown that if a mechanism $M$ fully specifies its cause purview, \emph{i.e.,} $\pi_c(Z=z^* \mid m ) = 1$, then any superset of that mechanism $\bar{M} \supseteq M$ also fully specifies the same purview, $\pi_c(Z\mid \bar{m}) = 1$. The same holds for the effect side as well. 
\begin{restatable}[]{Lemma}{supersetdetlem}
\label{lem:superset_det}
Superset of a deterministic mechanism is deterministic. For two mechanisms $M \subset \bar{M} \subseteq S$ and a single-unit purview $Z_i$,
If $\pi_e(Z_i=z_i \mid m )=1$, then $\pi_e(Z_i=z_i \mid \bar{m})=1$ and if $\pi_e(Z_i=z_i \mid m )=0$, then $\pi_e(Z_i=z_i \mid \bar{m})=0$.
\end{restatable}

Such constraints can be exploited to show that the upper bound can only be achieved by a subset of the mechanisms in the system. Another important fact to keep in mind is that $\varphi_e(m, Z=z) = |M||Z|$ is only achievable when $\pi_e(Z=z \mid m) = 1$. This is because this bound is derived for the informativeness term and is only achievable when the selectivity term, $\pi_e(Z=z \mid m)$, is $1$:
\begin{restatable}[]{Lemma}{pionenecessary}
\label{lem:pi_1_necessary}
If $\varphi_e(m, Z) = |M||Z|$ then $\pi_e(Z=z \mid m) = 1$. Similarly, if $\varphi_c(m, Z) = |M||Z|$ then $\pi_c(Z=z \mid m) = 1$. 
\end{restatable}
This is a very strict constraint if we want to construct a system such that $\varphi_e(m) = \varphi_e(m, z_e^*(m)) = |M||Z_e^*|$ for all the mechanisms. In what follows, we show that even having a single mechanism-purview pair that achieves this upper bound makes it impossible for a significant number of other mechanism-purview pairs to achieve their maximum integrated information. 
\begin{restatable}{Theorem}{effectinterlevel}
\label{thm:effect_interlevel}
$\varphi_e(\bar{m},\bar{Z}) < |\bar{M}||\bar{Z}|$, if $\varphi_e(m, Z) = |M||Z|$ and
\begin{itemize}
    \item $\bar{M} \subset M$ and $Z \cap \bar{Z} \neq \varnothing$, OR
    \item $\bar{M} \supset M$ and $Z \cap \bar{Z} \neq \varnothing$.
\end{itemize}
\end{restatable}
 Theorem \ref{thm:effect_interlevel} states that any subset or superset of $M$ cannot share purview units with $M$ and still achieve the maximum, if $\varphi_e(m ,Z) = |M||Z|$. Such mechanisms can only achieve their maximal integrated information over disjoint purviews.
This makes it immediately clear why all the distinctions in a system larger than $1$ unit cannot be maximally irreducible, as there are fewer disjoint purview sets than all the mechanisms.
As discussed in \nameref{S1_Appendix}, the same results hold for the integrated cause information as well. 
This makes the bounds in \eqref{eq:sum_bound_K_nonunique} and \eqref{eq:sum_bound_K} not achievable.
\begin{Corollary}
\label{cor:all_mech_effect}
All the distinctions in a system cannot be maximally integrated, if the system is composed of more than one unit.
\end{Corollary}
 This is in line with our intuitive definition of integrated information. $\varphi$ captures the difference that a mechanism makes over and beyond its parts. Therefore, if parts (subsets) of a mechanism-purview pair are maximally irreducible, the pair itself cannot be. 
For example, if there exist a mechanism in the system that achieves maximum $\varphi_e$ over the entire system, i.e. $|M||Z| = |M|N$, none of the subsets or supersets of that mechanism can achieve maximum $\varphi_e$ over their corresponding purviews. 
Another special case is when every single-unit mechanism achieves $\varphi_e=1$ over itself. Again, this makes it impossible for all the other distinctions to have maximum $\varphi_e$. In fact, in this case, the integrated effect information for the rest of the distinctions will be $0$, as any mechanism-purview pair is reducible to its parts.

\subsection{Intra-order constraints}
\label{subsec:selectivity_1}
As outlined in the preceding sections, having a high selectivity is necessary to have a large integrated information. For example, Lemma \ref{lem:pi_1_necessary} states that having selectivity of $1$ is necessary to achieve the maximal integrated effect and cause information.
This motivates us to study a special case that is particularly important. Assume that in a system consisting of $N$ units, all the mechanisms of size $K$ specify themselves with probability $1$, \emph{i.e.,} $\pi_e(Z=z' \mid m ) = 1$ where $Z = M$. 
Since the  mechanisms of size $K$ are not subsets or supersets of each other, Theorem \ref{thm:effect_interlevel} does not hold. However, even in this setting, the \emph{intra-order} constraints prevent the distinctions from achieving the maximum integrated information.
\begin{restatable}[]{Theorem}{orderkdet}
\label{thm:order_k_fully_det_effect}
In a system $S$ consisting of $N$ units, for a given mechanism size $1 < K < N$,
if $\pi_e(Z=z' \mid m ) = 1, \forall M: |M|=K$, 
and the purview units are the same as the mechanism units, \emph{i.e.,} $Z = M$, none of the mechanisms with $|M|=K$ can achieve their maximum integrated effect information of $|M||Z|=|M|^2$.
\end{restatable}
The theorem states that if all the mechanisms of size $K$ fully specify themselves, their $\varphi$ cannot achieve its maximum. The only exceptions are when there is only one such mechanism, $K = N$, or there is no overlap between the mechanisms $K = 1$. In \nameref{S1_Appendix}, we show that this result holds for the integrated cause information as well.

Analyzing this setup helps us to understand how mechanisms that are not parts of each other, but share parts, constrain each other.
The proof for Theorem \ref{thm:order_k_fully_det_effect}, which is provided in \nameref{S3_Appendix}, is a constructive proof. Thus, it also provides us with a procedure to construct a TPM that can achieve the maximum possible integrated information in this setting. The main idea behind the proof, and therefore the construction, is to use the definition of the effect repertoire in \eqref{eq:single_unit_pi_e} and \eqref{eq:pi_e} to characterize the TPM that can satisfy the assumptions of the theorem. This both gives us the TPM and limits the maximum attainable integrated effect information. 

In \nameref{S3_Appendix}, we also show that the MIP for the system under consideration can be found by evaluating only $\frac{K}{2} + 1$ partitions, significantly reducing the computations and making it feasible to calculate the integrated information of such distinctions in larger networks. In section \ref{sec:experiments}, we show that the integrated effect information in this system is much less than the number of connections by numerically evaluating the $\frac{K}{2} + 1$ candidate partitions. We can employ this linear time numerical evaluation to calculate the sum of integrated effect information over all the mechanisms in a hypothetical system where all the mechanisms of any size have themselves as the most irreducible purview with selectivity $1$ (not just all the mechanisms of size $K$). 
This gives us a numerical upper bound for reflexive systems, \emph{i.e.,} systems in which every mechanism has itself as the purview.
We will provide experiments to show how this computationally light numerical bound is tight and may in fact hold in general. 

\section{Upper bounds for relation integrated information}
\label{sec:relation_bounds}
Any mechanism $m$ with its maximally irreducible cause and effect purviews defines a causal distinction $d(m) = (z^*(m), \varphi_d(m) )$, where $
z^*(m) = \left( z^*_c(m), z^*_e(m) \right)$ contains the maximally irreducible cause and effect purviews in their maximal states. 

Causal relations are defined over subsets of the causal distinctions. 
If we define the set of all the valid distinctions of a candidate system as
$\mathcal{D}$,
a subset of distinctions $\boldsymbol{d} \subseteq \mathcal{D}$ forms a relation if the cause purview, the effect purview, or both the cause and the effect purviews of each distinction overlap over the same units in the same states\footnote{We use bold face to represent sets of distinctions, e.g., $\boldsymbol{d}$. $|\boldsymbol{d}|$ is the number of distinctions in the relation.}. This definition includes the special case of a self-relation where the cause and effect purviews of the same distinction overlap congruently over a set of units. This means a system consisting of $N$ units can potentially contain as many as $2^{2^N-1} - 1$ causal relations. Any nonempty subset of units can define a distinction and any nonempty subset of distinctions can potentially define a relation. Here we discuss the bound on the integrated information of a single relation and the sum of relations' integrated information, given a bound on the sum of distinctions' integrated information. This makes the analysis in this section mostly independent of the results in the previous sections, as the final results holds for any bound on the sum of distinctions' integrated information.

In \cite{Albantakis2023IntegratedTerms}, it is shown that we can write the relation integrated information of a set of distinctions $\boldsymbol{d} \subseteq \mathcal{D}$ with $|\boldsymbol{d}| \geq 2$ as:

\begin{equation}
\label{eq:relation_phi_d2}
\varphi_r(\boldsymbol{d}) =
\left|
\bigcap_{(z^*, \varphi_d ) \in \boldsymbol{d}} 
\big(z^*_c \cup z^*_e\big) 
\right | 
\min_{(z^*, \varphi_d ) \in \boldsymbol{d}} 
\frac{\varphi_d}{|z^*_c \cup z^*_e|}.
\end{equation}

Both the intersection and the union operators consider both the units in the purview and the state of units. If there is no congruent overlap among the purviews of the distinctions, $\bigcap_{(z^*, \varphi_d ) \in \boldsymbol{d}} \big(z^*_c \cup z^*_e\big)$ and  $\varphi_r(\boldsymbol{d})$ would be zero. For the special case of self-relations, $|\boldsymbol{d}| = 1$, the relation integrated information of a distinction $d$ is calculated as:
\begin{equation}
\label{eq:self_relation_phi}
\varphi_r(d) =
\left|
z^*_c \cap z^*_e
\right |  
\frac{\varphi_d}{|z^*_c \cup z^*_e|}.
\end{equation}
Our first observation is that the relation integrated information of any subset of distinctions $\boldsymbol{d}$ cannot be larger than the smallest distinction's integrated information
$$
\varphi_r(\boldsymbol{d}) \leq \min_{(z^*, \varphi_d ) \in \boldsymbol{d}} 
\varphi_d
$$
This is because the relation overlap
$
\left|
\bigcap_{(z^*, \varphi_d ) \in \boldsymbol{d}} 
\big(z^*_c \cup z^*_e\big) 
\right | 
$ is always smaller than the $|z^*_c \cup z^*_e|$ of all the distinctions involved. Therefore, to maximize the integrated information of an individual relation, we need to maximize the minimum integrated information of all the distinctions involved. In other words, we need to maximize the integrated information of all the distinctions.

However, as discussed in Section \ref{sec:mechanism_bounds}, due to inter-order and intra-order constraints, all the distinctions cannot achieve their maximum. 
Therefore, we study the problem of how we can maximize $\sum_{\boldsymbol{d} \subseteq \mathcal{D}} \varphi_r(\boldsymbol{d})$, given that $\sum_{d \in \mathcal{D}} \varphi_d$ is bounded by a certain value. 
We present the analysis in multiple steps:
\begin{enumerate}[label=(\roman*)]
    \item First, we rewrite the sum of the relations' integrated information $\sum_{\boldsymbol{d} \subseteq \mathcal{D}} \varphi_r(\boldsymbol{d})$ as a linear combination of distinctions integrated information,
    \item This turns our problem into a linear programming problem, as we are looking to maximize a linear combination of distinctions integrated information, given that the sum of them is bounded by some value. We then present the solution for this problem, which gives us the bound for $\sum_{\boldsymbol{d} \subseteq \mathcal{D}} \varphi_r(\boldsymbol{d})$ for any given set of cause and effect purviews. This provides us with a measure to compare different cause effect structures. We also briefly discuss the implications of this results on optimal connectivity profile,
    \item We finally use the bounds derived in the previous section to calculate the growth rate of $\sum_{\boldsymbol{d} \subseteq \mathcal{D}} \varphi_r(\boldsymbol{d})$ in terms of the number of units.
\end{enumerate}

\textbf{Step 1:} In \cite{Albantakis2023IntegratedTerms}, it is shown that we can write $\sum_{\boldsymbol{d} \subseteq \mathcal{D}} \varphi_r(\boldsymbol{d})$ as a linear combination of distinctions' integrated information by defining $\mathcal{Z}(o)$ as:
\begin{equation}
\label{eq:purview_inclusion_single_unit}
\mathcal{Z}(o) = \{ (z, \varphi) : z = z^*_c \cup z^*_e, ((z^*_c, z^*_e), \varphi) \in \mathcal{D}, o \in z^*_c \cup z^*_e   \}.
\end{equation}
Each element in $\mathcal{Z}(o)$ is a tuple that contains both $z^*_c \cup z^*_e$ and $\varphi$ of a distinction that contains $o$ in either its cause or effect purview. $o$ is an arbitrary unit in a arbitrary state in the system. 

This definition helps us to rewrite $\sum_{\boldsymbol{d} \subseteq \mathcal{D}} \varphi_r(\boldsymbol{d})$ as \cite[S3 Text]{Albantakis2023IntegratedTerms}: 
\begin{equation}
\label{eq:analytical_solution_relations}
\small
\sum_{\boldsymbol{d} \subseteq \mathcal{D}} \varphi_r(\boldsymbol{d})
= 
\sum_{\substack{\boldsymbol{d} \subseteq \mathcal{D}\\
                  |\boldsymbol{d}| = 1}}
                  \varphi_r(\boldsymbol{d})
+
\sum_{\substack{\boldsymbol{d} \subseteq \mathcal{D}\\
                  |\boldsymbol{d}| \geq 2}}
                  \varphi_r(\boldsymbol{d})
= 
\sum_{\substack{\boldsymbol{d} \subseteq \mathcal{D}\\
                  |\boldsymbol{d}| = 1}}
                  \varphi_r(\boldsymbol{d})
+
\sum_{o}
\sum_{i = 1}^{|\mathcal{Z}(o)|} \frac{\varphi_{(i)}}{|z_{(i)}|} (2^{|\mathcal{Z}(o)| - i} - 1).    
\end{equation}

The total sum is reformulated as the sum of the integrated information of self-relations and a linear combination of $\frac{\varphi_{(i)}}{|z_{(i)}|}$ of the distinctions. $\frac{\varphi_{(i)}}{|z_{(i)}|}$ denotes sorted indexing of the elements in $\mathcal{Z}(o)$, such that $(z_{(1)}, \varphi_{(1)})$ has the smallest $\frac{\varphi}{|z|}$ ratio, $(z_{(2)}, \varphi_{(2)})$ has the second smallest $\frac{\varphi}{|z|}$ ratio, and so on. The first summation in the last term sums over all the units $o$ in all their possible states.

To find an upper bound of $\sum_{\boldsymbol{d} \subseteq \mathcal{D}} \varphi_r(\boldsymbol{d})$ for any given set of cause and effect purviews, we can simplify the objective function by finding the maximum of each term in the sum separately. The maximization is over the values of the distinction integrated information:
\begin{equation}
\label{eq:sum_phi_r_bound}
\small
\sum_{\boldsymbol{d} \subseteq \mathcal{D}} \varphi_r(\boldsymbol{d})
\leq
\max 
\sum_{\substack{\boldsymbol{d} \subseteq \mathcal{D}\\
                  |\boldsymbol{d}| = 1}}
                  \varphi_r(\boldsymbol{d})
+
\sum_{o}
\max
\sum_{i = 1}^{|\mathcal{Z}(o)|} \frac{\varphi_{(i)}}{|z_{(i)}|} (2^{|\mathcal{Z}(o)| - i} - 1)
\end{equation}

\textbf{Step 2:} The first term is bounded by
$$
\sum_{\substack{\boldsymbol{d} \subseteq \mathcal{D}\\
                  |\boldsymbol{d}| = 1}}
                  \varphi_r(\boldsymbol{d})
\leq
\sum_{d \in \mathcal{D}} \varphi_d.
$$
This is because, as we just discussed, the integrated information of a self-relation is less than its distinction integrated information. 

We can find the maximum for the second term in the objective function by solving the following problem for a given $o$:
\begin{equation}
\label{eq:optimization}
\begin{aligned}
\max
&\sum_{i = 1}^{|\mathcal{Z}(o)|} \frac{\varphi_{(i)}}{|z_{(i)}|} (2^{|\mathcal{Z}(o)| - i} - 1), \\
\text{subject to}
& \sum_{i = 1}^{|\mathcal{Z}(o)|} \frac{\varphi_{(i)}}{|z_{(i)}|} \leq S(o)
\end{aligned}
\end{equation}
Even if the value of $S(o)$ is not known, we can find the optimal distribution of $\sum_{d \in \mathcal{D}} \varphi_d$ over all the distinctions, study the optimal connectivity profile, and analyze the trade-off between $S(o)$ and $|\mathcal{Z}(o)|$. 
In \nameref{S3_Appendix}, we provide the solution to this problem and the conditions to achieve the bound. In short, we show that a system can achieve the maximal value, if the distinctions with the same $|z^*_c \cup z^*_e|$ have the same value of integrated information. Systems with random connectivity cannot satisfy this symmetry and therefore they cannot have very large values of $\sum_{\boldsymbol{d} \subseteq \mathcal{D}} \varphi_r(\boldsymbol{d})$. On the other hand, the systems that exhibit homogeneous grid-like symmetries are more likely to satisfy it, as many subsets of units with same purview size have the same connectivity pattern, and therefore the same integrated information. 
We also show that the maximum achievable value for the objective function of \eqref{eq:optimization} is
\begin{equation}
\label{eq:phi_r_bound_o}
S(o) 
\left(
\frac{2^{|\mathcal{Z}(o)|}}{|\mathcal{Z}(o)|}
(1 - \frac{1}{2^{|\mathcal{Z}(o)|}}) - 1
\right).
\end{equation}
Plugging this result back into \eqref{eq:sum_phi_r_bound}, we get:
\begin{equation}
\label{eq:sum_phi_r_bound_final}
\begin{aligned}
\small
\sum_{\boldsymbol{d} \subseteq \mathcal{D}} \varphi_r(\boldsymbol{d})
\leq
\sum_{d \in \mathcal{D}} \varphi_d
+
\sum_{o}
S(o) 
\left(
\frac{2^{|\mathcal{Z}(o)|}}{|\mathcal{Z}(o)|}
(1 - \frac{1}{2^{|\mathcal{Z}(o)|}}) - 1
\right),
\end{aligned}
\end{equation}
which describes the trade-off between the number of distinctions that contain an specific purview unit $o$, $|\mathcal{Z}(o)|$, and the sum of their $\frac{\varphi}{|z|}$, $S(o)$. 
In general, due to the inter- and intra-order constraints, increasing 
$|\mathcal{Z}(o)|$ might decrease $S(o)$. However, since the bound increases almost exponentially in $|\mathcal{Z}(o)|$, in most cases we can increase the bound by increasing $|\mathcal{Z}(o)|$ for all $o$. This means that a system is more likely to achieve large values of $\sum_{\boldsymbol{d} \subseteq \mathcal{D}} \varphi_r(\boldsymbol{d})$, if its distinctions have bigger purviews. Therefore, densely connected grid-like systems are the candidate systems for high values of $\sum_{\boldsymbol{d} \subseteq \mathcal{D}} \varphi_r(\boldsymbol{d})$.

\textbf{Step 3:} To calculate a growth rate and a bound as a function of only the number of units in the system, we can study the following extreme case: We know that $|\mathcal{Z}(o)|$ cannot be larger than the total number of distinctions, \emph{i.e.}, $|\mathcal{Z}(o)| \leq 2^N - 1.$  
To find an upper bound for $S(o)$, we can use one of the bounds derived in Section \ref{subsec:single_mechanism}, \emph{i.e.,} $\varphi_e \leq |M||z_e|$. Since for any distinction $\varphi \leq \varphi_e$ and $|z^*_c \cup z^*_e| \geq |z^*_e|$, we have $\frac{\varphi}{|z^*_c \cup z^*_e|} \leq \frac{\varphi_e}{|z^*_e|} \leq |M|$. 
Therefore:
$$
\sum_{i = 1}^{|\mathcal{Z}(o)|} \frac{\varphi_{(i)}}{|z_{(i)}|} 
\leq
\sum_{|M|=1}^N |M| \binom{N}{|M|} = \frac{N}{2}2^N.
$$
Using the above values for $|\mathcal{Z}(o)|$ and $S(o)$, as well as the bound in \eqref{eq:sum_bound_K_nonunique}, we get:
\begin{equation}
\begin{aligned}
\small
\sum_{\boldsymbol{d} \subseteq \mathcal{D}} \varphi_r(\boldsymbol{d})
\leq
\frac{N^2}{2} 2^{N}
+
N^22^N 
\left(
\frac{2^{2^N - 1}}{2^N - 1}
(1 - \frac{1}{2^{2^N - 1}}) - 1
\right).
\end{aligned}
\end{equation}
Since achieving $\varphi_e = |M||z_e|$ is not possible for all the distinctions, this bound is not tight, but it gives us the growth rate of $\mathcal{O}(N^2 2^{2^{N}})$. Alternatively, we could have arrived at the same rate of growth by noticing that the maximum number of relations is $2^{2^N-1} - 1$ and the maximum value for distinction integrated information is $N^2$. 

The solution studied in this section translates the problem of finding a bound for $\sum_{\boldsymbol{d} \subseteq \mathcal{D}} \varphi_r(\boldsymbol{d})$ to the easier of problem finding a bound or a closed form solution for $S(o)$. In other words, Eq \eqref{eq:sum_phi_r_bound_final} readily gives us the bound for $\sum_{\boldsymbol{d} \subseteq \mathcal{D}} \varphi_r(\boldsymbol{d})$, for any given bound for $S(o)$.
Deriving a tighter general bound for $S(o)$, \emph{i.e.}, any arbitrary set of distinctions that share a common purview unit, leads to a tighter bound for $\sum_{\boldsymbol{d} \subseteq \mathcal{D}} \varphi_r(\boldsymbol{d})$ and remains an open problem.

\section{Numerical experiments}
\label{sec:experiments}

In this section, we provide numerical evaluations of the bounds discussed above, as well as experiments illustrating how tight the bounds are. The calculations were performed using the freely available PyPhi toolbox\footnote{ The code for PyPhi is available at \url{https://github.com/wmayner/pyphi}. The code for the experiments presented here is available at \url{https://github.com/zaeemzadeh/IIT-bounds}.} \cite{Mayner2018PyPhi:Theory}.
Since the values of mechanism integrated information and relation integrated information are state-specific, the bounds and the optimal TPM constructions are state-dependent. This means that if a system or mechanism is optimal in one state, it is not necessarily optimal in other states. Here, we show the results for the optimal state, as we are interested in the bounds and the maximal values.
Fig \ref{subfig:order_bound} shows the distinction integrated information of distinctions in the system discussed in Section \ref{subsec:selectivity_1}. The solid lines represent the integrated effect information, $\varphi_e^*(K)$, of a mechanism of size $K$ over itself, if all such mechanisms specify themselves with probability $1$. In this setting the mechanism size is same as the purview size. 
The dotted line represent the upper bound if we consider the mechanisms in isolation, which is simply the number of connections, \emph{i.e.,} $K^2$. For mechanism size of $K = N$ and $K=1$, where there is no overlap among the mechanisms, this upper bound is achievable. 
However, when the mechanisms share parts with each other, the integrated effect information cannot get close to the $K^2$ bound. 
For a fixed mechanism size $K$, if we increase the system size $N$, the maximum achievable $\varphi_e$ decreases. This is because in a larger system, there are more mechanisms of size $K$, making each mechanism compete with combinatorially more mechanisms. This figure also shows that for a fixed system size $N$, the maximum achievable $\varphi_e$ increases with the mechanism/purview size. However, this can be misleading. For example, although the mechanisms of size $K=N$ can achieve the largest $\varphi_e$ value, they are the least numerous mechanisms.  

\begin{figure}[h]
\centering    
\subfigure[ ]{
\label{subfig:order_bound}
\includegraphics[width=0.46\textwidth]{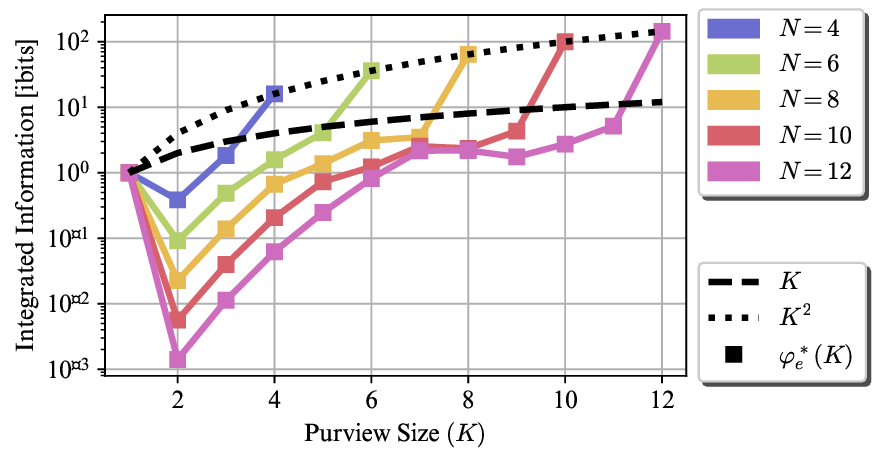}
}    
\subfigure[ ]{
\label{subfig:sum_order_bound}
\includegraphics[width=0.46\textwidth]{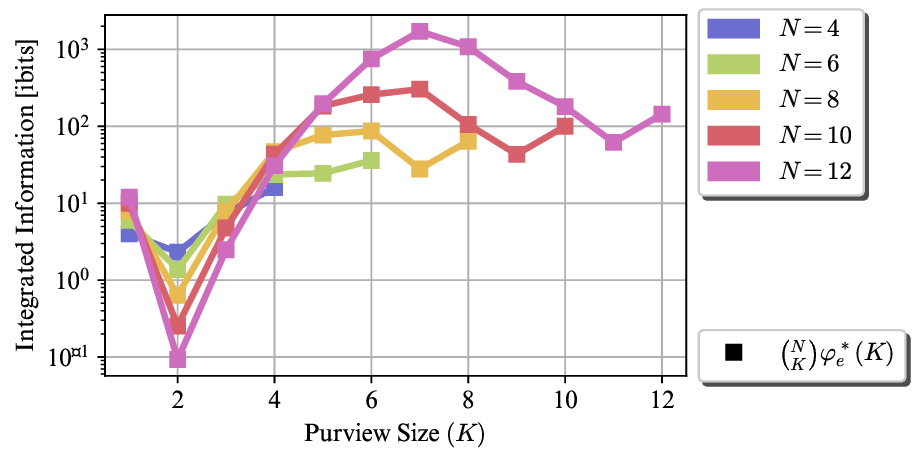}
}   
\caption{Upper bounds for integrated effect information $\varphi_e$ vs the value achieved by the proposed construction, for different mechanism/purview size $K=1, \dots, 12$. In this setting, the mechanism is same as the purview. The dotted line represents the upper bound if we consider each mechanism in isolation, which is the number of connections. The solid lines represent the integrated effect information achieved the proposed TPM construction, $\varphi_e^*(K)$. \subref{subfig:order_bound} $\varphi_e^*(K)$ for a single mechanism of size $K$ and \subref{subfig:sum_order_bound} $\sum \varphi_e^*(K)$ summed over all the mechanisms of size $K$ in a system of size $N$. This figure shows that the derived bound is achievable when the mechanisms do not share parts. Also, in larger systems, mechanisms of size close to $\frac{N}{2}$ achieve the maximum $\sum \varphi_e^*(K)$, as they are the most numerous ones.}
\end{figure}

To emphasize this point, Fig \ref{subfig:sum_order_bound} illustrates the sum of integrated effect information of all the mechanisms of size $K$, \emph{i.e.,} $\binom{N}{K} \varphi_e^*(K)$. This figure shows that although larger mechanisms can achieve higher $\varphi_e$, if the goal is to maximize $\sum \varphi_e$, it is better to maximize the integrated information of the more numerous mechanisms. 
For example, for the system size of $N=12$, the maximum $\binom{N}{K} \varphi_e^*(K)$ is achieved when $K=7$, which also illustrates the trade-off between the achievable $\varphi_e$ and the number of mechanisms. Although mechanisms of size $7$ are less numerous than mechanisms of size $6$, each of them can achieve  larger integrated effect information. 

Fig \ref{subfig:sum_order_bound} exhibits the maximum achievable $\sum \varphi_e(m)$ for a \emph{specific} mechanism size, when all such mechanisms fully specify themselves. Fig \ref{fig:sum_phi_vs_K} shows the sum of integrated information over all the mechanisms, if mechanisms of certain size fully specify themselves. For each system size $N$, $N$ different TPMs are constructed, each representing a system where all the mechanisms of size $K$ have selectivity of $1$, $K=1,\dots,N$. 
The dotted lines represent $\sum_{K=1}^N \binom{N}{K} \varphi_e^*(K)$. This value represents a hypothetical system in which all the mechanisms have selectivity of $1$ over themselves and can achieve their maximum $\varphi_e$. We can use this value as a numerical upper bound.

As shown in \nameref{S3_Appendix}, to evaluate $\varphi_e^*(K)$, we need to compute only $\frac{K}{2} + 1$ partitions, making the computational complexity of calculating $\sum_{K=1}^N \binom{N}{K} \varphi_e^*(K)$ quadratic in $N$. 
The solid lines in the figure is the achieved sum of integrated effect information $\sum_{M\subseteq S} \varphi_e(m)$, and the markers represent the achieved sum of integrated information $\sum_{M\subseteq S} \varphi(m) = \sum_{M\subseteq S} \min \{ \varphi_e(m), \varphi_c(m) \}$. It is evident that the sum bound is  tight, as $\sum_{M\subseteq S} \varphi_e(m)$ and $\sum_{M\subseteq S} \varphi(m)$ can get very close to the bound. 
This figure shows that we can exploit the symmetries and the constraints on the system to find good approximations of $\sum_{M\subseteq S} \varphi(m)$ with significantly less computation.

\begin{figure}[h]
\centering    
\includegraphics[width=0.75\textwidth]{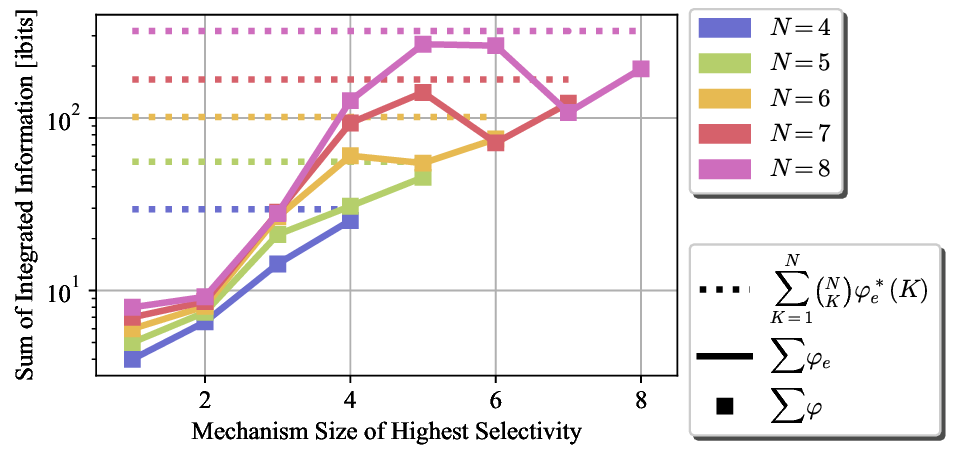}
\caption{Sum of integrated effect information $\sum_{M\subseteq S} \varphi_e(m)$ and sum of integrated information $\sum_{M\subseteq S} \varphi(m)$ for a system in which all the mechanisms of size $K$ specify themselves with probability $1$, \emph{i.e.,} selectivity is $1$, for different system sizes $N=4, \dots, 8$.  
In this setting, the mechanism is same as the purview. 
The dotted line is the numerical bound derived in Section \ref{subsec:selectivity_1} the solid line is $\sum_{M\subseteq S} \varphi_e(m)$ achieved by the proposed TPM construction, and markers denote $\sum_{M\subseteq S} \varphi(m)$ achieved by the construction. The bound is a tight upper bound for $\sum_{M\subseteq S} \varphi_e(m)$, and $\sum_{M\subseteq S} \varphi_e(m)$ is the same as $\sum_{M\subseteq S} \varphi(m)$. Furthermore, in larger systems, both of the sums are maximized when mechanisms of size close to $\frac{N}{2}$ have selectivity of $1$. }
\label{fig:sum_phi_vs_K}
\end{figure}

Table \ref{tab:individual_bounds} and Table \ref{tab:bounds} summarize the bounds discussed in this paper. Bound I and Bound II are the bounds derived in Eq \eqref{eq:sum_bound_K_nonunique} and Eq \eqref{eq:sum_bound_K}, respectively. They represent the absolute upper bounds without considering the inter- and intra-order constraints. To calculate their corresponding bound of the sum of relations' integrated information, the closed form sum in Eq \eqref{eq:analytical_solution_relations} is used, excluding the self-relation term for a less cluttered presentation. 
It is worthwhile to mention that each of the rows in Table \ref{tab:bounds} can be used to calculate its corresponding bound on $\Phi$, which is the sum of integrated information of all the distinctions and the relations. 

Bound III corresponds to the hypothetical system discussed in Section \ref{subsec:selectivity_1}, where all the mechanisms can achieve the maximal integrated effect information over themselves while having a selectivity of $1$. 
As discussed earlier, $\sum_{K=1}^N \binom{N}{K} \varphi_e^*(K)$ can be used as a tight upper bound for the sum of distinctions' integrated information in such system. 
To numerically calculate its corresponding bound for $\sum \varphi_r$, we further assumed $Z_c = S$ for all the distinctions and used Eq \eqref{eq:analytical_solution_relations}. This assumption makes the number of relations in this system the same as the number of relations in Bound I, as all the distinctions are related over their cause purviews. This makes the $\sum \varphi_r$ in Bounds I and III grow at a similar rate. Next, we will show Bound III is tight for both $\sum \varphi_d$ and $\sum \varphi_r$ of the proposed construction and may in fact be a general bound.
\begin{table}[!ht]
\begin{adjustwidth}{-2.25in}{0in} 
\centering
\caption{
{\bf The bounds for integrated information of an individual distinction, relation, and system}}
\small
\begin{tabular}{|m{8cm}|m{10cm}|}
\hline
 Description & Bound \\ \thickhline
Cause/effect integrated information of a mechanism-purview pair given a partition & 
 The number of causal connections severed by the partition, $\mathcal{N}(\theta)$.\\\hline
Cause/effect integrated information of a mechanism-purview & 
Total number of causal connections between the mechanism and the purview, $|M||Z|$. \\\hline
Relation integrated information of any subset of distinctions $\boldsymbol{d}$& 
The smallest distinction's integrated information,
$
\varphi_r(\boldsymbol{d}) \leq \min_{(z^*, \varphi_d ) \in \boldsymbol{d}} 
\varphi_d
$.\\\hline
System integrated information of a system with $|S|$ units& 
The maximum number of connections cut by any valid cut, which is $|S|(|S|-1)$ for the system partitions considered in \cite{Marshall2023SystemInformation}. \\\hline
\end{tabular}
\label{tab:individual_bounds}
\end{adjustwidth}
\end{table}

\begin{table}[!ht]
\begin{adjustwidth}{-2.25in}{0in} 
\centering
\caption{
{\bf The bounds for the sum of distinctions integrated information $\sum \varphi_d$ and the sum of relations integrated information $\sum \varphi_r$.}}
\small
\begin{tabular}{|c|c|l|l|l|l|}
\hline
 & Description & $\sum \varphi_d$ & $\mathcal{O}(\sum \varphi_d)$ & $\sum \varphi_r$ & $\mathcal{O}(\sum \varphi_r)$\\ \thickhline
 \makecell{ Bound I \\ (not achievable)} & 
 \makecell{$\forall M \subseteq S, \varphi_d = |M||S|$ \\ $Z_e^* = Z_c^* = S$} &
 $\frac{N^2}{2}2^N$ &
 $N^2 2^N$ &
 \footnotesize{ $N \sum_K K \left( 2^{\sum_{j = K}^N\binom{N}{j}} ( 1 - 0.5^{\binom{N}{K}}) -  \binom{N}{K}\right)$}
 & 
 $N^2 2^{2^N}$
 \\ \hline
 Bound II & 
  \makecell{$\forall M \subseteq S, \varphi_d = |M|^2$ \\ $Z_e^* = Z_c^* = M$} &
 $\frac{N(N+1)}{4}2^N$ &
 $N^2 2^N$ &
 \footnotesize{$N \sum_K K \left( 2^{\sum_{j = K-1}^{N-1}\binom{N-1}{j}} ( 1 - 0.5^{\binom{N-1}{K-1}}) -  \binom{N-1}{K-1}\right)$} & 
 $N^2 2^{2^{N-1}}$
 \\ \hline
 Bound III &  
\makecell{The numerical bound\\ derived in Section \ref{subsec:selectivity_1}}
&
 N/A &
 $2^N$ &
N/A & 
 $2^{2^{N}}$
 \\ \hline
\end{tabular}
\label{tab:bounds}
\end{adjustwidth}
\end{table}

Fig \ref{fig:sum_phi_vs_N_random} illustrates how the upper bounds and the maximum achievable $\sum_{M\subseteq S} \varphi(m)$ by different TPM construction methods grow as we increase the system size $N$. The TPMs are constructed by the following methods:
\begin{enumerate*}[label=(\roman*)]
    \item \textbf{High selectivity} refers to the construction discussed in Section \ref{subsec:selectivity_1} and Theorem \ref{thm:order_k_fully_det_effect}. As shown in \nameref{S3_Appendix}, the TPM corresponding to this construction only contains $0$s and $1$s.
    \item \textbf{Random deterministic} construction fills the TPM entries with only $0$s and $1$s at random. The probability of an entry being one is drawn from a uniform distribution over $[0,1]$. 
    \item \textbf{Random} construction draws the entries of the TPM from a uniform distribution over $[0,1]$.
    \item \textbf{Hamming} is the decoding procedure of the famous $(7,4)$ Hamming code. In a system of $7$ units, the state at time $t+1$ is the closest valid codeword to the state at time $t$ with probability $1$. This construction has been suggested as a good candidate system for high integrated information \cite{Tegmark2015ConsciousnessMatter}.
\end{enumerate*}

Fig \ref{fig:sum_phi_vs_N_random} shows that both the numerical upper bound and the achievable values for the proposed construction scale almost exponentially with the system size $N$ and at the rate of $\mathcal{O}(2^N)$. 
For both of the random constructions, the  maximum sum achieved over $100$ runs is reported in the figure. This figure shows that generally deterministic TPMs achieve higher sum of integrated information, compared to nondeterministic TPMs. This is because to achieve large $\varphi$, selectivity close to $1$ is necessary (see Lemma \ref{lem:pi_1_necessary}) and to have selectivity of $1$, TPM entries need to be $1$ or $0$ (see Lemma \ref{lem:det_mech}). 
This figure also suggests that the numerical upper bound derived in Section \ref{subsec:selectivity_1}, \emph{i.e.,} $\sum_{K} \binom{N}{K} \varphi_e^*(K)$, might be a general upper bound for any system, not just for high selectivity regime. The generality of this upper bound remains to be investigated. 

\begin{figure}[h] 
\centering    
\includegraphics[width=1.\textwidth]{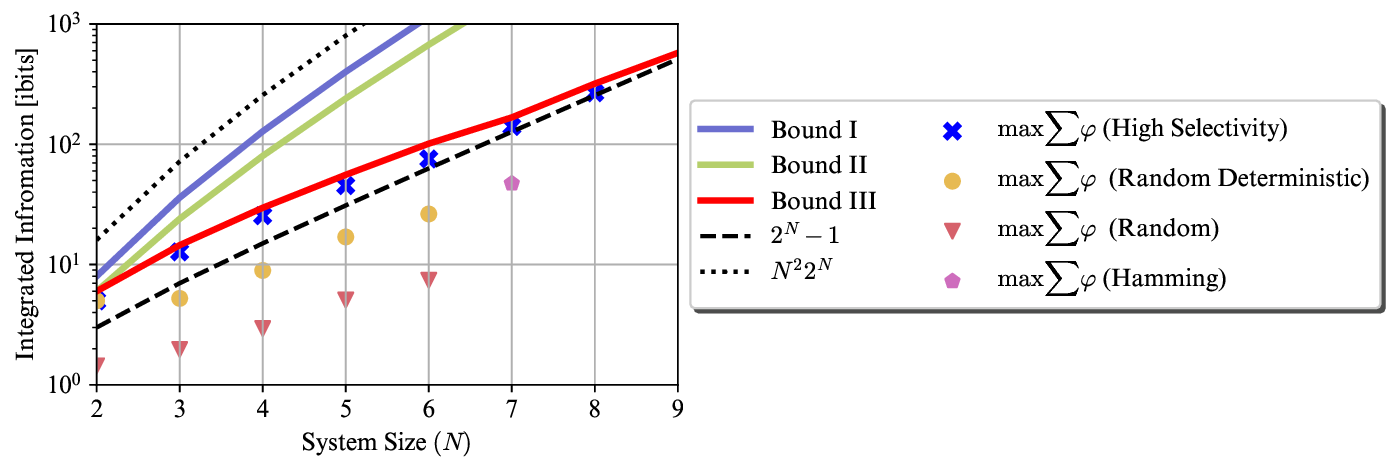}
\caption{Growth of different upper bounds for the sum of distinctions' integrated information w.r.t the system size $N$, as well as the maximum achievable $\sum_{M\subseteq S} \varphi(m)$ by different TPM construction methods. 
The derived numerical upper bound grows almost exactly as $2^N$. $\sum_{M\subseteq S} \varphi(m)$ achievable by the proposed construction grows exponentially with $N$ as well. 
The deterministic TPMs, \emph{i.e.,} TPMs containing only $0$s and $1$s, outperform nondeterministic TPMs. The scale of the $y$-axis is logarithmic.
}
\label{fig:sum_phi_vs_N_random}
\end{figure}

Finally, Fig \ref{fig:sum_phi_relations_vs_N_random} illustrates the bounds and the achievable values for the sum of integrated information of relations in a system. Similar to the distinctions case the proposed construction can get very close to its bound. Furthermore, this figure shows that higher sum of distinctions integrated information does not translate to higher sum of relations integrated information. For example, unlike Fig \ref{fig:sum_phi_vs_N_random},  bound III is larger than bound II in Fig \ref{fig:sum_phi_relations_vs_N_random}. This is because having larger purviews increases the number of relations exponentially, making the purview sizes, not the integrated information value, the dominant factor in the sum of relations integrated information. Due to the super exponential growth of the number of the relations in terms of the number of the units, we used double-logarithmic scale for the $y$-axis.

\begin{figure}[h] 
\centering    
\includegraphics[width=1.\textwidth]{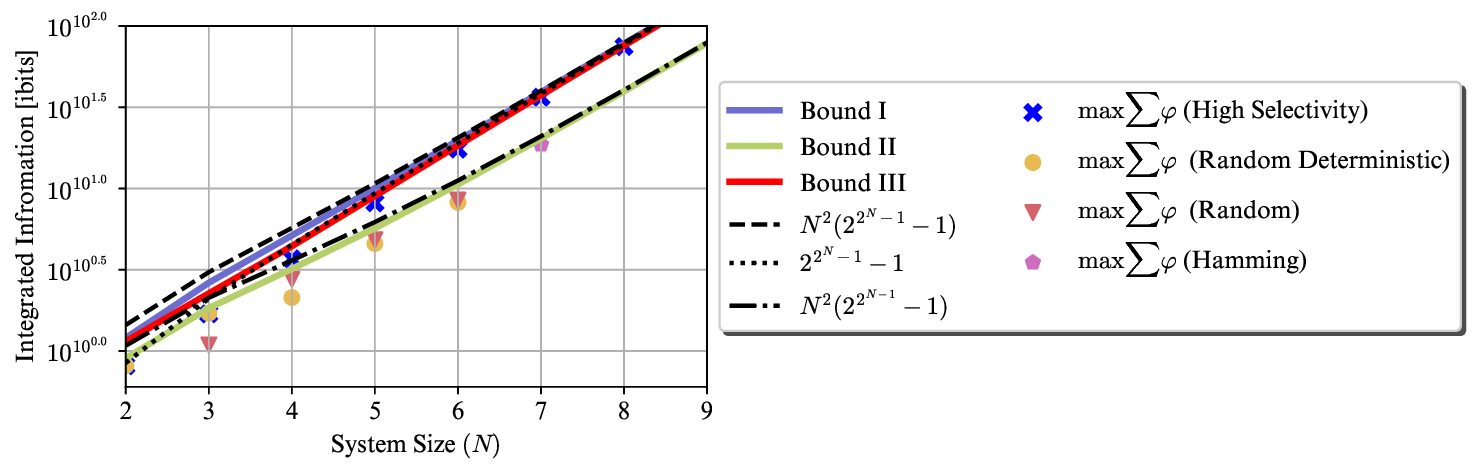}
\caption{Growth of different upper bounds for the sum of relations' integrated information w.r.t the system size $N$, as well as the maximum achievable values by different TPM construction methods. 
The derived numerical upper bound and  $\sum_{M\subseteq S} \varphi(m)$ achievable by the proposed construction grow super exponentially with $N$. The scale of the $y$-axis is double-logarithmic.
}
\label{fig:sum_phi_relations_vs_N_random}
\end{figure}

\section{Discussion and future directions}
\label{sec:discussion}
As a theory of consciousness, Integrated Information Theory (IIT) \cite{Albantakis2023IntegratedTerms, Tononi2016IntegratedSubstrate} introduces a mathematical framework that quantifies the causal powers of systems and subsets of units in a system. 
The methods of IIT have inspired activity in the neuroscience \cite{Casali2013ABehavior, Haun2017ConsciousElectrocorticography, Gallimore2015RestructuringTheory, Lee2008PropofolAnesthesia, Leung2021IntegratedMelanogaster}, complexity science \cite{Albantakis2014EvolutionComplexity, Albantakis2021QuantifyingMeasures, Mediano2022IntegratedComplexity, Tajima2017IntegratedDynamics}, and physics \cite{Tegmark2015ConsciousnessMatter, Barrett2014AnPhysics, Albantakis2023ComputingMechanism}. 
In short, IIT identifies the essential properties of experience and provides a framework to quantify to what extent a physical system complies with such properties in terms of its causal powers. 
These properties are intrinsicality, information, integration, exclusion, and composition \cite{Albantakis2023IntegratedTerms}.
This work is mainly focused on the last property, \emph{i.e.,} composition, 
meaning that we are concerned with 
how the subsets of units in the system specify causes and effects over subsets of units and how these causes and effects are related. 
Our results are applicable even if a candidate system does not satisfy all the essential properties.

Since the exact calculation of integrated information in larger biological systems is not feasible, theoretical investigation of different connectivity profiles or approximate calculations are necessary. The results and the proofs presented in this work are a step toward practical application of IIT to real-world systems in several ways. 
First, we provide the growth rate of several measures as a function of the number of units under certain assumptions. Such results and techniques can be used to derive analytical forms for other connectivity profiles. 

Second, we provide techniques to exploit the symmetries and the repeating patterns in connectivity to simplify the computation of the integrated information in certain families of systems. Similar techniques can be used to derive approximate measures for any biological system with repeating patterns of connectivity. For instance, In \nameref{S3_Appendix}, we show that the normalized partitioned integrated information can be simplified to average information gain over all the severed connections. This can potentially be used in conjunction with exact or approximate subset selection algorithms to find the optimal partition in a more computationally feasible manner. In this work, we used this technique and the symmetries in the system to reduce the number of candidate partitions significantly, from more than exponential to linear in size of the mechanism.

Finally, our results provide us with qualitative insights about the connectivity profile of the biological systems that can achieve higher values of integrated information.
For example in the proof of Theorem \ref{thm:order_k_fully_det_effect}, we provide a TPM construction method that maximizes the sum of distinction integrated information under certain assumptions. Similar techniques can be used to design optimal TPMs and study biological systems satisfying other assumptions. Our analysis can also be helpful to compare different families of networks, \emph{e.g.}, random vs homogeneous connectivity. This can lead to better understanding of which areas of brain can have higher integrated information. For example, in Section \ref{sec:relation_bounds}, we derive a bound for the sum of relations integrated information for any set of cause and effect purviews. 
Our results showed that the sum of integrated information of relations is maximized when the distinctions' integrated information are proportional to the size of their purviews. This means distinctions with the same size of purview should have the same value of integrated information. Such regularity is not attainable in systems with random connectivity profile, but can potentially be achieved in systems with homogeneous grid-like connectivity. 

Theoretical investigation of any measure can help us to sharpen our intuition and deepen our understanding of that measure. For example, optimal coding theorems establish the connection between the Shannon entropy measure and the minimum  number of bits assigned to symbols. This makes Shannon information theory more tangible and easier to understand. Similarly, in this paper, we show that the value of mechanism integrated information, system integrated information, and the information loss by any partition are closely related to the number of causal connections. This gives us a more intuitive understanding of their behavior.
More importantly, the results discussed here are utilized in the development of IIT. For example, in IIT, to fairly compare partitions that severe different number of causal connections, the value of the information loss corresponding to each partition is normalized by the maximum information loss achievable by such partition. The formal proof for the upper bound achievable by any partition is presented in Section \ref{subsec:single_mechanism}. 

More generally, our results have implications outside the field of theoretical neuroscience, such as finding the signatures of complexity in biological and artificial networks. There has been increasing interest in IIT as a formal framework to study complex systems and non-linear dynamics \cite{Niizato2020FindingTheory, Tajima2017IntegratedDynamics, Albantakis2014EvolutionComplexity, Albantakis2021QuantifyingMeasures, Mediano2022IntegratedComplexity}. Therefore, our results and analysis can also be extended to the field of complexity science and to potentially derive approximate or heuristic measures of complexity \cite{Oizumi2016UnifiedGeometry, Kanwal2017ComparingMachines, Nilsen2019EvaluatingInformation}. In \nameref{S2_Appendix}, we show that the results provided here are still relevant even if other distance measures, such as point-wise mutual information or KL divergence, are used for quantifying the integrated information.

There are also many questions left open to be explored. The results provided here are state-dependent, meaning that they are achievable only in one state of the system and/or the mechanism. 
The problem of finding the conditions under which the systems and the mechanisms can achieve high integrated information in more than one state is open to investigation. 
Also, as shown in Section \ref{sec:mechanism_bounds}, we can make the distinction bound tighter by considering the inter- and intra-order constraints. In Section \ref{subsec:selectivity_1}, we developed a linear time method to evaluate the integrated information in a special class of systems and numerically showed that the this bound is much smaller than the bound without considering the constraints. 
Finding a closed form solution for such systems remains an open problem. In general, finding a closed form solution for the bound on sum of the integrated information of any arbitrary subset of distinctions is an interesting problem and can be used in conjunction with our results in Section \ref{sec:relation_bounds} to derive tighter bounds for sum of integrated information of relations.

Furthermore, in our derivation of the bounds we mainly focused on the integrated effect information, as the bound for $\varphi_e$ is a bound for $\varphi_d$ as well. In \nameref{S1_Appendix}, we show that most of our results are applicable to the integrated cause information as well. However, there are dependencies among the cause and effect side that can potentially be used to make the overall bounds tighter.

\section*{Supporting information}

\paragraph*{S1 Appendix.}
\label{S1_Appendix}
{\bf Integrated cause information.} 
Definition of the integrated cause information and description of how the results for integrated effect information can be translated to the integrated cause information.

\paragraph*{S2 Appendix.}
\label{S2_Appendix}
{\bf Using other difference measures.} 
Describes how the results change if other difference measures, such as KL divergence, are used for comparing the distributions.

\paragraph*{S3 Appendix.}
\label{S3_Appendix}
{\bf Proofs} 

\section*{Acknowledgments}
The authors thank Larissa Albantakis, William Marshall, Leonardo Barbosa, and Erick Chastain for helpful discussions and comments on earlier drafts of this article, and Will Mayner for maintaining PyPhi, his advice on implementing the experiments, and his helpful comments.
In addition, this research was supported by the Templeton World Charity Foundation (grant number TWCF0216) and by the Tiny Blue Dot Foundation (UW 133AAG3451; G.T.).

\nolinenumbers

\bibliography{references.bib, more_ref}

\begin{thebibliography}{10}

\bibitem{Albantakis2023IntegratedTerms}
Albantakis L, Barbosa L, Findlay G, Grasso M, Haun AM, Marshall W, et~al.
\newblock {Integrated information theory (IIT) 4.0: Formulating the properties
  of phenomenal existence in physical terms}.
\newblock PLOS Computational Biology. 2023;19(10):e1011465.
\newblock doi:{10.1371/JOURNAL.PCBI.1011465}.

\bibitem{Tononi2016IntegratedSubstrate}
Tononi G, Boly M, Massimini M, Koch C.
\newblock {Integrated information theory: from consciousness to its physical
  substrate}.
\newblock Nature Reviews Neuroscience 2016 17:7. 2016;17(7):450--461.
\newblock doi:{10.1038/nrn.2016.44}.

\bibitem{Massimini2005BreakdownSleep}
Massimini M, Ferrarelli F, Huber R, Esser SK, Singh H, Tononi G.
\newblock {Breakdown of cortical effective connectivity during sleep}.
\newblock Science. 2005;309(5744):2228--2232.
\newblock doi:{10.1126/SCIENCE.1117256/SUPPL{\_}FILE/MASSIMINI.SOM.PDF}.

\bibitem{Pigorini2015BistabilitySleep}
Pigorini A, Sarasso S, Proserpio P, Szymanski C, Arnulfo G, Casarotto S, et~al.
\newblock {Bistability breaks-off deterministic responses to intracortical
  stimulation during non-REM sleep}.
\newblock NeuroImage. 2015;112.
\newblock doi:{10.1016/j.neuroimage.2015.02.056}.

\bibitem{Haun2019WhyExperience}
Haun A, Tononi G.
\newblock {Why Does Space Feel the Way it Does? Towards a Principled Account of
  Spatial Experience}.
\newblock Entropy 2019, Vol 21, Page 1160. 2019;21(12):1160.
\newblock doi:{10.3390/E21121160}.

\bibitem{Comolatti_Time}
{Comolatti, R and Grasso, M and Tononi, G}.
\newblock {Why does time feel the way it does? }; forthcoming.

\bibitem{Casarotto2016StratificationComplexity}
Casarotto S, Comanducci A, Rosanova M, Sarasso S, Fecchio M, Napolitani M,
  et~al.
\newblock {Stratification of unresponsive patients by an independently
  validated index of brain complexity}.
\newblock Annals of Neurology. 2016;80(5):718--729.
\newblock doi:{10.1002/ANA.24779}.

\bibitem{Sarasso2020LocalInjury}
Sarasso S, D'Ambrosio S, Fecchio M, Casarotto S, Vigan{\`{o}} A, Landi C,
  et~al.
\newblock {Local sleep-like cortical reactivity in the awake brain after focal
  injury}.
\newblock Brain. 2020;143(12):3672--3684.
\newblock doi:{10.1093/BRAIN/AWAA338}.

\bibitem{Findlay2019}
Findlay G, Marshall W, Albantakis L, Mayner WGP, Koch C, Tononi G.
\newblock {Dissociating Intelligence from Consciousness in Artificial Systems
  – Implications of Integrated Information Theory}.
\newblock In: Proceedings of the 2019 Towards Conscious AI Systems Symposium,
  AAAI SSS19; 2019.

\bibitem{Tononi2015Consciousness:Everywhere}
Tononi G, Koch C.
\newblock {Consciousness: here, there and everywhere?}
\newblock Philosophical Transactions of the Royal Society B: Biological
  Sciences. 2015;370(1668).
\newblock doi:{10.1098/RSTB.2014.0167}.

\bibitem{Barbosa2021MechanismInformation}
Barbosa LS, Marshall W, Albantakis L, Tononi G.
\newblock {Mechanism Integrated Information}.
\newblock Entropy 2021, Vol 23, Page 362. 2021;23(3):362.
\newblock doi:{10.3390/E23030362}.

\bibitem{Barbosa2020AInformation}
Barbosa LS, Marshall W, Streipert S, Albantakis L, Tononi G.
\newblock {A measure for intrinsic information}.
\newblock Scientific Reports 2020 10:1. 2020;10(1):1--9.
\newblock doi:{10.1038/s41598-020-75943-4}.

\bibitem{Marshall2023SystemInformation}
Marshall W, Grasso M, Mayner WGP, Zaeemzadeh A, Barbosa LS, Chastain E, et~al.
\newblock {System Integrated Information}.
\newblock Entropy 2023, Vol 25, Page 334. 2023;25(2):334.
\newblock doi:{10.3390/E25020334}.

\bibitem{Mayner2018PyPhi:Theory}
Mayner WGP, Marshall W, Albantakis L, Findlay G, Marchman R, Tononi G.
\newblock {PyPhi: A toolbox for integrated information theory}.
\newblock PLOS Computational Biology. 2018;14(7):e1006343.
\newblock doi:{10.1371/JOURNAL.PCBI.1006343}.

\bibitem{Tegmark2015ConsciousnessMatter}
Tegmark M.
\newblock {Consciousness as a state of matter}.
\newblock Chaos, Solitons {\&} Fractals. 2015;76:238--270.
\newblock doi:{10.1016/J.CHAOS.2015.03.014}.

\bibitem{Casali2013ABehavior}
Casali AG, Gosseries O, Rosanova M, Boly M, Sarasso S, Casali KR, et~al.
\newblock {A theoretically based index of consciousness independent of sensory
  processing and behavior}.
\newblock Science Translational Medicine. 2013;5(198).
\newblock
  doi:{10.1126/SCITRANSLMED.3006294/SUPPL{\_}FILE/5-198RA105{\_}SM.PDF}.

\bibitem{Haun2017ConsciousElectrocorticography}
Haun AM, Oizumi M, Kovach CK, Kawasaki H, Oya H, Howard MA, et~al.
\newblock {Conscious Perception as Integrated Information Patterns in Human
  Electrocorticography}.
\newblock eNeuro. 2017;4(5).
\newblock doi:{10.1523/ENEURO.0085-17.2017}.

\bibitem{Gallimore2015RestructuringTheory}
Gallimore AR.
\newblock {Restructuring consciousness -The psychedelic state in light of
  integrated information theory}.
\newblock Frontiers in Human Neuroscience. 2015;9(June):346.
\newblock doi:{10.3389/FNHUM.2015.00346/BIBTEX}.

\bibitem{Lee2008PropofolAnesthesia}
Lee U, Kim S, Noh GJ, Choi BM, Mashour GA.
\newblock {Propofol Induction Reduces the Capacity for Neural Information
  Integration: Implications for the Mechanism of Consciousness and General
  Anesthesia}.
\newblock Nature Precedings 2008. 2008; p. 1--1.
\newblock doi:{10.1038/npre.2008.1244.2}.

\bibitem{Leung2021IntegratedMelanogaster}
Leung A, Cohen D, Van~Swinderen B, Tsuchiya N.
\newblock {Integrated information structure collapses with anesthetic loss of
  conscious arousal in Drosophila melanogaster}.
\newblock PLOS Computational Biology. 2021;17(2):e1008722.
\newblock doi:{10.1371/JOURNAL.PCBI.1008722}.

\bibitem{Albantakis2014EvolutionComplexity}
Albantakis L, Hintze A, Koch C, Adami C, Tononi G.
\newblock {Evolution of Integrated Causal Structures in Animats Exposed to
  Environments of Increasing Complexity}.
\newblock PLOS Computational Biology. 2014;10(12):e1003966.
\newblock doi:{10.1371/JOURNAL.PCBI.1003966}.

\bibitem{Albantakis2021QuantifyingMeasures}
Albantakis L.
\newblock {Quantifying the Autonomy of Structurally Diverse Automata: A
  Comparison of Candidate Measures}.
\newblock Entropy 2021, Vol 23, Page 1415. 2021;23(11):1415.
\newblock doi:{10.3390/E23111415}.

\bibitem{Mediano2022IntegratedComplexity}
Mediano PAM, Rosas FE, Farah JC, Shanahan M, Bor D, Barrett AB.
\newblock {Integrated information as a common signature of dynamical and
  information-processing complexity}.
\newblock Chaos. 2022;32(1):13115.
\newblock doi:{10.1063/5.0063384/2835635}.

\bibitem{Tajima2017IntegratedDynamics}
Tajima S, Kanai R.
\newblock {Integrated information and dimensionality in continuous attractor
  dynamics}.
\newblock Neuroscience of Consciousness. 2017;2017(1).
\newblock doi:{10.1093/NC/NIX011}.

\bibitem{Barrett2014AnPhysics}
Barrett AB.
\newblock {An integration of integrated information theory with fundamental
  physics}.
\newblock Frontiers in Psychology. 2014;5(FEB):63.
\newblock doi:{10.3389/FPSYG.2014.00063/BIBTEX}.

\bibitem{Albantakis2023ComputingMechanism}
Albantakis L, Prentner R, Durham I.
\newblock {Computing the Integrated Information of a Quantum Mechanism}.
\newblock Entropy 2023, Vol 25, Page 449. 2023;25(3):449.
\newblock doi:{10.3390/E25030449}.

\bibitem{Niizato2020FindingTheory}
Niizato T, Sakamoto K, Mototake Yi, Murakami H, Tomaru T, Hoshika T, et~al.
\newblock {Finding continuity and discontinuity in fish schools via integrated
  information theory}.
\newblock PLOS ONE. 2020;15(2):e0229573.
\newblock doi:{10.1371/JOURNAL.PONE.0229573}.

\bibitem{Oizumi2016UnifiedGeometry}
Oizumi M, Tsuchiya N, Amari SI.
\newblock {Unified framework for information integration based on information
  geometry}.
\newblock Proceedings of the National Academy of Sciences of the United States
  of America. 2016;113(51):14817--14822.
\newblock doi:{10.1073/PNAS.1603583113/-/DCSUPPLEMENTAL}.

\bibitem{Kanwal2017ComparingMachines}
Kanwal MS, Grochow JA, Ay N.
\newblock {Comparing Information-Theoretic Measures of Complexity in Boltzmann
  Machines}.
\newblock Entropy 2017, Vol 19, Page 310. 2017;19(7):310.
\newblock doi:{10.3390/E19070310}.

\bibitem{Nilsen2019EvaluatingInformation}
Nilsen AS, Juel BE, Marshall W.
\newblock {Evaluating Approximations and Heuristic Measures of Integrated
  Information}.
\newblock Entropy 2019, Vol 21, Page 525. 2019;21(5):525.
\newblock doi:{10.3390/E21050525}.

\bibitem{Albantakis2019WhatNetworks}
Albantakis L, Marshall W, Hoel E, Tononi G.
\newblock {What Caused What? A Quantitative Account of Actual Causation Using
  Dynamical Causal Networks}.
\newblock Entropy 2019, Vol 21, Page 459. 2019;21(5):459.
\newblock doi:{10.3390/E21050459}.

\end{thebibliography}
\label{LastPage}

\appendix
\newpage
\fancyfoot{}
\section[\appendixname~\thesection]{Integrated cause information}
\label{app:cause}
In Section \ref{sec:mechanism_bounds}, we discussed the definition and the upper bounds of the distinction integrated information. For simplicity, our discussion and theoretical results were mostly focused on the integrated effect information. Since $\varphi_d = \min \{ \varphi_c, \varphi_e \} \leq \varphi_e$, any bound on $\varphi_e$ or $\sum \varphi_e$ is a bound for $\varphi_d$ or $\sum \varphi_d$. Here, for completeness, we discuss the definition of integrated cause information and show how the results obtained for the effect side can be translated to the cause side. 

The cause repertoire $\pi_c(Z_{t-1} \mid M = m)$ is defined as as the probability distribution over a potential cause purview at time $t-1$, $Z_{t-1}$, given a mechanism in state $m$ and can be calculated using Bayes' rule:
$$
\pi_c(z \mid m) = \frac{\pi_e(m \mid z) \pi_c(z)}{\pi_e(m;Z)},
$$
where $\pi_e(m;Z)$ is the unconstrained effect probability as defined in Section \ref{sec:mechanism_bounds} and $\pi_c(z) = |\Omega_Z|^{-1}$ is the unconstrained cause probability. The time subscripts are dropped to avoid cluttering the notation. Similar to the effect side, the maximal cause state of the mechanism $m$ over a potential cause purview $Z$ can be found as:
$$
    z'_c(m, Z) =\arg \max_{z \in \Omega_Z} \pi_c(z \mid m)\log_2\left(\frac{\pi_e(m \mid z)}{\pi_e(m;Z)}\right).
$$
Since there is at least one state with $\pi_e(m \mid z) \geq \pi_e(m;Z)$, the maximal cause state is always a state for which the cause $z$ increases the probability of the mechanism $m$ compared to its unconstrained probability. Finally, the integrated cause information is calculated as:
\begin{equation}
    \label{eq:purview_phi_c_def}
    \varphi_c(m,Z) = \pi_c(z'_c \mid m)
    \pospart{\log_2\left(\frac{\pi_e(m \mid z'_c)}{\pi_e^{\theta'}(m \mid z'_c)}\right)},
\end{equation}
where $\theta'$ is the minimum information partition (MIP), the partition that achieves the minimum normalized integrated cause information. 
We can also find the most irreducible cause purview of a mechanism as:
$$
z_c^*(m) = \arg\max_{\{z'_c | Z \subseteq S\}} \varphi_c(m, Z = z'_c).
$$
These steps are similar to the steps described for the effect side in Section \ref{sec:mechanism_bounds}. The only difference is the definition of the repertoires and the use of a slightly different distance measure. The measure being used for the cause side has the form $\small \pi_c(z \mid m)\pospart{\log_2\left(\frac{\pi_e(m \mid z)}{\pi_e'(m \mid z)}\right)}$, where the selectivity uses the backward probability and the informativeness uses the forward probability. This measure satisfies the following properties:
\begin{enumerate*}[label=(\roman*)]
    \item The measure differs from $0$ only if the cause state increases the probability of the mechanism state.
    \item The measure is not an aggregate over all the states and reflects how much change is made in an individual state.
    \item In a scenario where adding more units to the cause purview does not increase the probability of the mechanism state further, having the new units in the maximally irreducible purview is discouraged.
\end{enumerate*}

More importantly, since the informativeness term is the same for both $\varphi_e$ and $\varphi_c$, all the results described for $\varphi_e$ hold for $\varphi_c$ as well. For example, the normalization factor derived in Lemma \ref{lem:num_connections} holds for both the cause and effect information, as it is a bound on the infromativeness term. Similarly, since the proofs for Theorem \ref{thm:effect_interlevel} and Theorem \ref{thm:order_k_fully_det_effect} only use the informativeness term to derive a bound for the integrated effect information, they both hold for the integrated cause information.

\newpage
\section[\appendixname~\thesection]{Using other difference measures}
\label{app:other_measures}
As discussed in Section \ref{sec:mechanism_bounds} and \nameref{S1_Appendix}, the difference measure to quantify $\varphi_e$ and $\varphi_c$ satisfy certain properties. In this section, we study how our results might change if we change the difference measure to satisfy slightly different properties. 

\subsection[\appendixname~\thesection]{Absolute informativeness}
Due to using the positive part operator $|.|_+$ in the informativeness, integrated information is non-zero only if the state at $t-1$ increases the probability of the state at $t$. This matches our intuition that a cause should increase the probability of its effect. However, we can show even if we use the absolute value operator in the difference measure, our results would stay the same, by showing that the upper bound for the negative case is smaller than the upper bound for the positive case.
Lemma \ref{lem:neg_log_ratio} states the upper bound for the case where the mechanism \emph{decreases} the probability.

\begin{restatable}[Maximum for negative cause/effect]{Lemma}{negativecause}
\label{lem:neg_log_ratio}
For $0 \leq p, q \leq 1$, if $p \leq q$ then $$
p\absval{\log_2(\frac{p}{q})} \leq q\frac{\log_2(e)}{e} \leq  \frac{\log_2(e)}{e} \approx 0.531.
$$
\end{restatable}


\begin{proof}
For $0 \leq p \leq q \leq 1$, we have $p\absval{\log_2(\frac{p}{q})} = - p\log_2(\frac{p}{q})$ and 
$$
\begin{aligned}
\frac{d(- p\log_2(\frac{p}{q}))}{dp} &= \frac{-\ln(\frac{p}{q}) - 1}{\ln(2)}.\\
\frac{d^2(- p\log_2(\frac{p}{q}))}{dp^2} &= -\frac{1}{\ln(2)p} < 0.
\end{aligned}
$$
Therefore, $p\absval{\log_2(\frac{p}{q})}$ is concave for $0 \leq p \leq q \leq 1$.  The maximum is obtained by setting the first derivative to $0$:
$$
\begin{aligned}
\frac{d(- p\log_2(\frac{p}{q}))}{dp} &= \frac{-\ln(\frac{p}{q}) - 1}{\ln(2)} = 0\\
p = \frac{q}{e}
\end{aligned}
$$
For $p = \frac{q}{e}$, we have $p\absval{\log_2(\frac{p}{q})} = \frac{q}{e}\log_2(e) \leq \frac{\log_2(e)}{e}$. Thus for any $p \leq q$, $p\absval{\log_2(\frac{p}{q})} \leq \frac{\log_2(e)}{e}$.
\end{proof}


In other words, achieving integrated information of larger than $\frac{\log_2(e)}{e} \approx 0.531$ is not possible by decreasing the probability. Therefore, the mechanism needs to increase the probability to achieve large values of integrated information. All the bounds discussed in Section \ref{sec:mechanism_bounds} are larger than this value, therefore changing $|.|_+$ to $|.|$ does not change our results.

\subsection[\appendixname~\thesection]{Point-wise mutual information}
Another candidate for the distance measure is the point-wise mutual information, \emph{i.e.,} $\pospart{\log(\frac{p}{q})}$. Similar to the measure used in the definition of the integrated information, this measure is not an aggregate over the states and quantifies the change in an individual state. Furthermore, it differs from $0$ only if $p > q$. This measure has been used previously to quantify actual causation is discrete systems \cite{Albantakis2019WhatNetworks}. Using this measure does not change our results because the bounds derived in Section \ref{sec:mechanism_bounds} were all derived for the informativeness term of difference measure, which has the form $\pospart{\log(\frac{p}{q})}$. 

\subsection[\appendixname~\thesection]{Kullback–Leibler divergence}
Kullback–Leibler divergence (KLD) measure is another common distance measure for probability distribution and is defined as $D_{KL}(p(X) || q(X)) = \sum_{x \in \Omega_X} p(X=x) \log(\frac{p(X=x)}{q(X=x)})$. It is $0$, only if $p(X=x) = q(X=x), \forall x \in \Omega_X$, and is an aggregate measure over all the states, unlike our default measure.

Here, we study the upper bound for $D_{KL}(\pi_e(Z \mid m) || \pi_e^{\theta}(Z \mid m))$ and show that the results in Section \ref{sec:mechanism_bounds} are still relevant even if we use KLD. First, using the definition of the effect repertoire and additivity of KLD, we have:

\begin{adjustwidth}{-2.25in}{0cm}
\begin{equation}
\label{eq:KLD_decomp}
D_{KL}(\pi_e(Z \mid m) || \pi_e^{\theta}(Z \mid m)) = 
D_{KL}(\prod_{Z_i \in Z}\pi_e(Z_i \mid m ) || \prod_{Z_i \in Z} \pi_e^{\theta}(Z_i \mid m )) =
\sum_{Z_i \in Z} D_{KL}(\pi_e(Z_i \mid m ) || \pi_e^{\theta}(Z_i \mid m )).
\end{equation}
\end{adjustwidth}
Let us define the mechanism connected to $Z_i$ after the partitioning as $M_i$. For each individual binary unit in the purview $Z_i$, we can use the proof in lemma \ref{lem:pi_ratio} to show that:
$$
\begin{aligned}
\frac{\pi_e(Z_i= z_i  \mid m)}{\pi_e(Z_i = z_i \mid m_i)} \leq  2^{|M| - |M_i| }.
\end{aligned}
$$
Therefore, finding the maximum value for an individual binary unit boils down to solving the following problem:
$$
\begin{aligned}
\max_{p,q} \quad 
& D_{KL}(p || q) = p\log(\frac{p}{q}) + (1-p)\log(\frac{1-p}{1-q})\\
\text{subject to} \quad  
& p \leq 2^{|M| - |M_i| } q \\
& q \leq p \leq 1.
\end{aligned}
$$
The constraint $q \leq p$ is added without loss of generality, as it always holds for one of the two states. We are denoting the probabilities of that state with $p$ and $q$. Under this constraint the derivative of the objective function with respect to $p$ is always non-negative:
$$
\frac{\partial D_{KL}(p || q)}{\partial p} = 
\log\left(\frac{p}{q}\right) - \log\left(\frac{1-p}{1-q}\right) =
\log\left(\frac{\frac{1}{q} - 1}{\frac{1}{p} - 1}\right) \geq 0
$$
This means the maximum of $D_{KL}(p || q)$ occurs at the boundary for which $p$ is the largest, which is $p = \min\{1, 2^{|M| - |M_i| } q\}$. In other words, $p = 2^{|M| - |M_i| } q$ for $q \leq \frac{1}{2}^{|M| - |M_i| }$ and $p = 1$, otherwise. For the latter case, we have $D_{KL}(p || q) = \log(\frac{1}{q})$, whose maximum occurs when $q$ is minimized, \emph{i.e.,} $q = \frac{1}{2}^{|M| - |M_i| }$, with the maximal KLD value of $|M| - |M_i|$. The maximum for the former case also happens at the same point, since the derivative of the objective function with respect to $q$ is non-negative:
$$
\frac{\partial}{\partial q} 
\left(
 2^{|M| - |M_i| } q (|M| - |M_i|)
 +
 (1-2^{|M| - |M_i| } q)\log(\frac{1-2^{|M| - |M_i| } q}{1-q})
\right)
\geq 0,
$$
and we need to maximize $q$ in order to maximize $D_{KL}(p || q)$. This again leads to the maximum value of $|M| - |M_i|$ at $q = \frac{1}{2}^{|M| - |M_i| }$.
Thus, we have:
$$
D_{KL}(\pi_e(Z_i \mid m ) || \pi_e(Z_i \mid m _i))
\leq |M| - |M_i|,
$$
which is the same bound as Lemma \ref{lem:pi_ratio}, but for KLD instead of our default distance measure. $|M| - |M_i|$ is the number of connection severed  form $Z_i$. We can plug this result into \eqref{eq:KLD_decomp} and arrive at:
$$
D_{KL}(\pi_e(Z \mid m) || \pi_e^{\theta}(Z \mid m)) = 
\sum_{Z_i \in Z} D_{KL}(\pi_e(Z_i \mid m ) || \pi_e^{\theta}(Z_i \mid m ))
\leq \sum_{Z_i \in Z} (|M| - |M_i|)
= \mathcal{N}(\theta).
$$
$\mathcal{N}(\theta)$ is the number of connections severed  by the partition $\theta$. This is the same result as Lemma \ref{lem:num_connections} and shows that even if we use KLD as the distance measure, the normalization factor for finding the MIP would not change. This can also be used to show that if we use KLD to quantify $\varphi_e$, it cannot be larger than the number of connections between the mechanism and the purview, \emph{i.e.,} $\varphi_e(m, E) \leq |M||E|$ (same results as Theorem \ref{thm:single_mechanism}).

Furthermore, since KLD is maximized when $\pi_e(Z \mid m)$ is fully deterministic, Lemma \ref{lem:pi_1_necessary} and Theorem \ref{thm:effect_interlevel} hold. Finally, the assumption for Theorem \ref{thm:order_k_fully_det_effect} is that the effect repertoire is fully deterministic, \emph{i.e.,} $\pi_e(Z=z \mid m) = 1$ for some $z \in \Omega_Z$. In this case, KLD is simplified to $\log(\frac{1}{\pi_e^{\theta}(Z=z \mid m)}) = \pi_e(Z=z \mid m)\pospart{\log(\frac{\pi_e(Z=z \mid m)}{\pi_e^{\theta}(Z=z \mid m)})}$, which coincides with our default distance measure. Therefore, Theorem \ref{thm:order_k_fully_det_effect} holds as well.

\newpage
\section[\appendixname~\thesection]{Proofs}
\label{app:proofs}
\subsection[\appendixname~\thesection]{Single mechanism}
Before presenting our main proofs, let us revisit the definition of single-unit effect repertoire. Given the set of units outside the mechanism $W = S-M$ and a single unit $Z_i$ in state $z_i$:
$$
\pi_e(Z_i = z_i \mid m ) = \frac{1}{2^{|W|}} \sum_w p(Z_i  = z_i \mid m , w) = \frac{1}{|\mathcal{M}(m)|} \sum_{s \in \mathcal{M}(m)} p(Z_i = z_i \mid s),
$$
where $\mathcal{M}(m) \subset \Omega_S$ is the set of system states in which mechanism $M$ is in state $m$ and we have $| \mathcal{M}(m) | = 2^{|W|} = 2^{N-|M|}$. $\Omega_S$ is the set of all possible states of the system $S$. Defining $\mathcal{M}(m)$ helps us to present our results in a more concise manner. For example, we will use the following simple lemma to prove Lemma \ref{lem:pi_ratio}

\begin{Lemma}
\label{lem:M_subset}
For two mechanisms $M$ and $\bar{M}$ such that $ \bar{M}\subset M \subseteq S$, we have $  \mathcal{M}(m) \subset \mathcal{M}(\bar{m})$.
\end{Lemma}
\begin{proof}
Without loss of generality let us assume $M$ is in all-zero state. By definition, $\forall s \in \mathcal{M}(m)$ the units corresponding to $M$ are in all-zero state. Since $\bar{M} \subset M$, $\forall s \in \mathcal{M}(m)$ the units corresponding to $\bar{M}$ are also in all-zero state and therefore $s \in \mathcal{M}(\bar{m})$. Thus, $\forall s \in \mathcal{M}(m)$, we have $s \in \mathcal{M}(\bar{m})$, which means $  \mathcal{M}(m) \subset \mathcal{M}(\bar{m})$.
\end{proof}

\begin{Example}
$\bar{M} = A$ and $M=AB$ for system $S = ABC$ in all-zero state. $\mathcal{M}(\bar{m}) = \{ ABC=000, ABC=001, ABC=010, ABC=011\}$ and $\mathcal{M}(m) = \{ ABC=000, ABC=001\}$.
\end{Example}

\piratiolem*
\begin{proof}
For a single purview unit $Z_i \in Z$ we have:
$$
\begin{aligned}
\frac{\pi_e(Z_i = z_i \mid m)}{\pi_e(Z_i = z_i \mid \bar{m})} = \frac{\frac{1}{|\mathcal{M}(m)|} \sum_{s \in \mathcal{M}(m)}p(Z_i = z_i \mid s)}{\frac{1}{|\mathcal{M}(\bar{m})|} \sum_{s \in \mathcal{M}(\bar{m})}p(Z_i= z_i  \mid s)}
\end{aligned}
$$
Since $\bar{M} \subset M$, we have $  \mathcal{M}(m) \subset \mathcal{M}(\bar{m})$ (using Lemma \ref{lem:M_subset}), therefore $\sum_{s \in \mathcal{M}(m)}p(Z_i = z_i \mid s) \leq \sum_{s \in \mathcal{M}(\bar{m})}p(Z_i = z_i \mid s)$. This is because we are taking the sum over more terms on the right side. Thus
$$
\begin{aligned}
\frac{\pi_e(Z_i= z_i  \mid m)}{\pi_e(Z_i = z_i \mid \bar{m})} \leq \frac{\frac{1}{|\mathcal{M}(m)|}} {\frac{1}{|\mathcal{M}(\bar{m})|}} = \frac{\frac{1}{2^{N-|M|}} }{\frac{1}{2^{N-|\bar{M}|}} } = 2^{|M| - |\bar{M}| }.
\end{aligned}
$$
\end{proof}
\numconnectionslem*
\begin{proof}
Due to the conditional independence of the purview units, we have:
$$
\log_2(\frac{\pi_e(Z = z \mid m )}{\pi_e^\theta(Z = z \mid m )}) 
= \log_2(\frac{\prod_i \pi_e(Z_i = z_i \mid m)}{\prod_i \pi_e^\theta(Z_i = z_i \mid m)}) 
= \sum_{Z_i \in Z} \log_2(\frac{\pi_e(Z_i = z_i \mid m)}{\pi_e^\theta(Z_i = z_i \mid m)}) 
$$
Let us define $M_i$ as the mechanism connected to $Z_i$ after the partitioning. Using Lemma \ref{lem:pi_ratio}:
$$
\sum_{Z_i \in Z} \log_2(\frac{\pi_e(Z_i = z_i \mid m)}{\pi_e^\theta(Z_i = z_i \mid m)}) 
\leq \sum_{Z_i \in Z} (|M| - |M_i|) = \mathcal{N}(\theta).
$$
$\mathcal{N}(\theta) = \sum_{Z_i \in Z} (|M| - |M_i|)$ is the total number of connections cut by partition $\theta$.
\end{proof}

\singlemechthm*
\begin{proof}
For the effect side, we have:
$$
\small
\begin{aligned}
\varphi_e(m, E) 
&=  \pi_e(E = e \mid m )  
\pospart{ \log_2(\frac{\pi_e(E = e \mid m )}{\pi_e^{\theta'}(E = e \mid m )}) }
\leq \pospart{ \log_2(\frac{\pi_e(E = e \mid m )}{\pi_e^{\theta'}(E = e \mid m )}) } \\
&\leq \max_\theta
\pospart{ \log_2(\frac{\pi_e(E = e \mid m )}{\pi_e^{\theta}(E = e \mid m )}) }
\overset{(a)}{\leq} \max_\theta \mathcal{N}(\theta)
\leq |M||E|
\end{aligned}
$$
$\theta'$ represents the MIP and equality (a) follows from Lemma \ref{lem:num_connections}.
Similarly for the cause side, we have:
$$
\small
\begin{aligned}
\varphi_c(m, C) 
&=  \pi_c(C = c \mid m )  
\pospart{ \log_2(\frac{\pi_e(M = m \mid c)}{\pi_e^{\theta'}(M = m \mid c)}) }
\leq \pospart{ \log_2(\frac{\pi_e(M = m \mid c)}{\pi_e^{\theta'}(M = m \mid c)}) } \\
&\leq \max_\theta
\pospart{ \log_2(\frac{\pi_e(M = m \mid c)}{\pi_e^{\theta}(M = m \mid c)}) }
\leq \max_\theta \mathcal{N}(\theta)
\leq |M||C|
\end{aligned}
$$
\end{proof}

\pionenecessary*
\begin{proof}
In deriving the upper bounds in Theorem \ref{thm:single_mechanism}, the upper bound is achieved by setting  $\pi_e(E = e \mid m )$ and $\pi_c(C = c \mid m )$ to $1$. Using the same line of proof and setting $\pi_e(E = e \mid m ) < 1$ and $\pi_c(C = c \mid m ) < 1$, we will achieve $\varphi_e(m, E) < |M||E|$ and $\varphi_c(m, C) < |M||C|$. Thus, to achieve $\varphi_e(m, E) = |M||E|$ ($\varphi_c(m, C) = |M||C|$), we need to have $\pi_e(E = e \mid m ) = 1$ ($\pi_c(C = c \mid m ) = 1$).
\end{proof}

\subsection[\appendixname~\thesection]{Inter-order constraints}

To prove Theroem \ref{thm:effect_interlevel}, we first need to prove a few intermediate lemmas. 
\begin{Lemma}[Deterministic mechanism]
\label{lem:det_mech}
For a mechanism $M \subseteq S$ and a single-unit purview $Z_i$, if $\pi_e(Z_i=z_i \mid m )=1$, then $p(Z_i=z_i \mid s)=1, \forall s \in \mathcal{M}(m)$. Furthermore, if $\pi_e(Z_i=z_i \mid m )=0$, then $p(Z_i=z_i \mid s)=0, \forall s \in \mathcal{M}(m)$.
\end{Lemma}
\begin{proof}

For $\pi_e(Z_i=z_i \mid m )=1$:
$$
\pi_e(Z_i=z_i \mid m ) = \frac{1}{|\mathcal{M}(m)|} \sum_{s \in \mathcal{M}(m)} p(Z_i = z_i \mid s) = 1.
$$
Therefore, $\sum_{s \in \mathcal{M}(m)} p(Z_i = z_i \mid s) = |\mathcal{M}(m)|$. Since $p(Z_i = z_i \mid s) \leq 1, \forall s$, all the terms in the summation need to be $1$, \emph{i.e.,} $p(Z_i=z_i \mid s)=1, \forall s \in \mathcal{M}(m)$.

For $\pi_e(Z_i=z_i \mid m )=0$, $\sum_{s \in \mathcal{M}(m)} p(Z_i = z_i \mid s) = 0$. Since $p(Z_i = z_i \mid s) \geq 0, \forall s$, all the terms in the summation need to be $0$, \emph{i.e.,} $p(Z_i=z_i \mid s)=0, \forall s \in \mathcal{M}(m)$.
\end{proof}

\supersetdetlem*
\begin{proof}
If $\pi_e(Z_i=z_i \mid m )=1$:

According to Lemma \ref{lem:det_mech}, $p(Z_i=z_i \mid s)=1, \forall s \in \mathcal{M}(m)$. Since $M \subset \bar{M}$, we have $\mathcal{M}(\bar{m}) \subset \mathcal{M}(m)$ (Lemma \ref{lem:M_subset}). Therefore, $p(Z_i=z_i \mid s)=1, \forall s \in \mathcal{M}(\bar{m})$ and $\pi_e(Z_i=z_i \mid \bar{m}) = \frac{1}{|\mathcal{M}(\bar{m})|} \sum_{s \in \mathcal{M}(\bar{m})} p(Z_i = z_i \mid s) = 1$.

Similarly, for $\pi_e(Z_i=z_i \mid m )=0$: 

According to Lemma \ref{lem:det_mech}, $p(Z_i=z_i \mid s)=0, \forall s \in \mathcal{M}(m)$. Since $M \subset \bar{M}$, we have $\mathcal{M}(\bar{m}) \subset \mathcal{M}(m)$ (Lemma \ref{lem:M_subset}). Therefore, $p(Z_i=z_i \mid s)=0, \forall s \in \mathcal{M}(\bar{m})$ and $\pi_e(Z_i=z_i \mid \bar{m}) = \frac{1}{|\mathcal{M}(\bar{m})|} \sum_{s \in \mathcal{M}(\bar{m})} p(Z_i = z_i \mid s) = 0$.
\end{proof}

\begin{Lemma}
\label{lem:inter_level_upward}
If $\varphi_e(m, Z) = |M||Z|$ then $\varphi_e(\bar{m}, \bar{Z}) < |\bar{M}||\bar{Z}|, $ if $ M \subset \bar{M}$ and $Z \cap \bar{Z} \neq \varnothing$.
\end{Lemma}
\begin{proof}
Since $\varphi_e(m, Z) = |M||Z|$, using Lemma \ref{lem:pi_1_necessary}:
$$
\small
\begin{aligned}
&\pi_e(Z=z^* \mid m ) = \prod_{Z_i \in Z} \pi_e(Z_i=z_i^* \mid m ) = 1 \\
\implies
&\pi_e(Z_i=z_i^* \mid m ) = 1, \forall Z_i \in Z\\
\overset{(a)}{\implies} 
& p(Z_i=z_i \mid s)=1, \forall s \in \mathcal{M}(m), \forall Z_i \in Z\\
\implies
& \pi_e(Z_i=z_i^*) = \frac{1}{|\Omega_S|} \sum_{s \in \Omega_S} p(Z_i = z_i \mid s)
\geq \frac{|\mathcal{M}(m)|}{|\Omega_S|}
= \frac{2^{N - |M|}}{2^N} = \frac{1}{2^{|M|}}, \forall Z_i \in Z, 
\end{aligned}
$$

From Theorem \ref{thm:single_mechanism}, we already know $\varphi_e(\bar{m}, \bar{Z}) \leq |\bar{M}||\bar{Z}|$ and from Lemma \ref{lem:num_connections}, we know the only partitioning that can achieve $|\bar{M}||\bar{Z}|$ is the complete partition, \emph{i.e.,} removing all the connections between $\bar{M}$ and $\bar{Z}$.
Now, we show even the complete partition cannot achieve $|\bar{M}||\bar{Z}|$. Under the complete partition, $\theta = \{ (\varnothing, \bar{M}), (\bar{Z}, \varnothing) \}$, we have:

\begin{adjustwidth}{-2.25in}{0cm}
$$
\small
\begin{aligned}
\varphi_e(\bar{m}, \bar{Z}, \theta)
& = \pi_e(\bar{Z} = \bar{z} \mid \bar{m}) 
\pospart{ \log_2 (\frac{\pi_e(\bar{Z} = \bar{z} \mid \bar{m})}{\pi_e(\bar{Z} = \bar{z})}) }
\leq 
\pospart{ \log_2 (\frac{\pi_e(\bar{Z} = \bar{z} \mid \bar{m})}{\pi_e(\bar{Z} = \bar{z})}) } \\
&= 
\pospart{ \sum_{Z_i \in \bar{Z}} \log_2 (\frac{\pi_e(\bar{Z}_i = \bar{z}_i \mid \bar{m})}{\pi_e(\bar{Z}_i = \bar{z}_i)}) }
\leq 
\pospart{ \sum_{Z_i \in \bar{Z} \cap Z} \log_2 (\frac{\pi_e(\bar{Z}_i = \bar{z}_i \mid \bar{m})}{\pi_e(\bar{Z}_i = \bar{z}_i)}) }
+ 
\pospart{ \sum_{Z_i \in \bar{Z} - Z} \log_2 (\frac{\pi_e(\bar{Z}_i = \bar{z}_i \mid \bar{m})}{\pi_e(\bar{Z}_i = \bar{z}_i)}) } \\
&\overset{(a)}{\leq} |\bar{Z} \cap Z||M| + |\bar{Z} - Z||\bar{M}| < |\bar{Z} \cap Z||\bar{M}| + |\bar{Z} - Z||\bar{M}| = |\bar{Z}||\bar{M}|
\end{aligned}
$$
\end{adjustwidth}
The first term in (a) follows from the fact that $\pi_e(Z_i=z_i^*) \geq \frac{1}{2^{|M|}}, \forall Z_i \in Z$, as we proved earlier, and the second term follows from Lemma \ref{lem:pi_ratio}.
\end{proof}

\begin{Lemma}
\label{lem:inter_level_downward}
$\varphi_e(\bar{m}, \bar{Z}) < |\bar{M}||\bar{Z}|, $ if $\varphi_e(m, Z) = |M||Z|$ and $\bar{M} \subset M$ and $Z \cap \bar{Z} \neq \varnothing$.
\end{Lemma}
\begin{proof}
Proof by contradiction: Assume $\varphi(\bar{m}, \bar{Z}) = |\bar{M}||\bar{Z}|$. Then according to Lemma \ref{lem:inter_level_upward}, $\varphi(m, Z) < |M||Z|$. Contradiction.
\end{proof}

\effectinterlevel*
\begin{proof}
Follows directly from Lemma \ref{lem:inter_level_downward} and Lemma \ref{lem:inter_level_upward}.
\end{proof}

\subsection[\appendixname~\thesection]{Intra-order constraints}
\label{app:intra-proofs}
In this section, we study a system consisting of $N$ units in which all the mechanisms of size $K$ specify themselves with probability $1$:
$$
\pi_e(Z = z' \mid m) = 1, \forall M: |M|=K, Z=M.
$$
First, we show how we can construct a TPM that can satisfy this property. Then, we show why such mechanisms cannot achieve their maximal value. Finally, we show, for such mechanism, MIP can only be one of a few candidate partitions.

Let us denote $M_{n,j}$ and $Z_{n,j}$, $j=1,\dots, \binom{N-1}{K-1}$ and $n=1,\dots,N$ as the set of mechanisms and purviews of size $K$ that contain $n^{\text{th}}$ unit, respectively. From the assumptions of the theorem we also have $Z_{n,j} = M_{n,j}, |Z_{n,j}| = |M_{n,j}| = K, \forall j,n$. Furthermore, $Z_n$ denotes the single-unit purview only containing the $n^{\text{th}}$ unit.
Starting from the assumption of the theorem, we have:
$$
\begin{aligned}
\pi_e(Z_{n,j}=z_{n,j}^* \mid m _{n,j}) &= 1, \forall n,j \implies 
\pi_e(Z_n=z_n^* \mid m _{n,j}) &= 1, \forall n,j
\end{aligned}
$$
This expression means the probability of the $n^{\text{th}}$ purview unit, given all the mechanisms of size $K$ that contain the $n^{\text{th}}$ unit is $1$.
Using Lemma \ref{lem:det_mech}:
$$
p(Z_n=z_n^* \mid s) = 1, \forall s \in \mathcal{M}(m_{n,j}), n = 1,\dots, N, \forall j
$$
Note that in general $z_n^*$, which is the maximal state the purview unit, depends on the mechanism $M_{n,j}$. But the above constraint is only satisfiable when $z_n^*$ is the same for all $j$. To see that, first notice that $\mathcal{M}(m_{n,j}) \cap \mathcal{M}(m_{n,j'}) \neq \varnothing, \forall j,j'$ and for $s \in \mathcal{M}(m_{n,j}) \cap \mathcal{M}(m_{n,j'})$, either we have $p(Z_n=0 \mid s) = 1$ or $p(Z_n=1 \mid s) = 1$. Therefore, $M_{n,j}$ and $M_{n,j'}$ should agree on the state of $Z_n$. This is true for any pair $j,j'$. Without loss of generality, we assume $z_n^*$ is $0$ for all $n$ and $j$.
\begin{equation}
    \label{eq:construction}
    p(Z_n=0 \mid s) = 1,  \forall s \in \bigcup_j \mathcal{M}(m_{n,j}), n = 1,\dots, N.
\end{equation}
Furthermore, without loss of generality, let us assume that the system is also in all-zero state at the current time.
Thus, $\mathcal{M}(m_{n,j})$ is the set of system states in which the state of the units in $M_{n,j}$ (unit $n$ and $K-1$ other units) are $0$ and $\bigcup_j \mathcal{M}(m_{n,j})$ is the set of all the system states for which the state of the unit $n$ and at least $K - 1$ other units are $0$.

So far, we have shown that, for any single purview unit $Z_n$, the TPM entry for $p(Z_n=0 \mid s)$ is $1$ for all system states $s$ that the unit $n$ and at least $K - 1$ other units are $0$. To study the partitioned repertoires, we should also derive $\pi_e(Z_n=0\mid \bar{m})$ for different mechanisms $\bar{m} \subset m_{n,j}$ with $|\bar{M}| < K$.

First, consider the case where $\bar{M}$ contains unit $n$ and $|\bar{M}| < K$:
$$
\pi_e(Z_n=0 \mid \bar{m} = 0) = \frac{1}{|\mathcal{M}(\bar{m})|} 
\sum_{s \in \mathcal{M}(\bar{m})} p(Z_n = 0 \mid s),
$$
Let us denote this probability as $\pi(|\bar{M}|)$. $\mathcal{M}(\bar{m})$ is the set of system states for which unit $n$ and $|\bar{M}| - 1$ other units are $0$ and we are marginalizing over the state of the rest of $N - |\bar{M}|$ units. But as we discussed earlier, we know any state with unit $n$ and at least $K - 1$ other units in state $0$ has $p(Z_n = 0 \mid s) = 1$. So out of the $2^{N - |\bar{M}|}$ states in $\mathcal{M}(\bar{m})$, there are at least $\sum_{b = K - |\bar{M}|}^{N - |\bar{M}|} \binom{N - |\bar{M}|}{b}$ states with probability $1$. Therefore:
$$
\pi(|\bar{M}|) \geq \frac{\sum_{b = K - |\bar{M}|}^{N - |\bar{M}|} \binom{N - |\bar{M}|}{b}}{2^{N - |\bar{M}|}}.
$$
Similarly, we can define the probability $\pi_e(Z_n=0 \mid \bar{m} = 0)$ for the case where $\bar{M}$ does not contain unit $n$ and $ |\bar{M}| < K$ as $\bar{\pi}(|\bar{M}|)$. In this case, $\mathcal{M}(\bar{m})$ is the set of system states for which $|\bar{M}|$ units are $0$ and we are marginalizing over the state of the rest of $N - |\bar{M}|$ units. Again, any state in $\mathcal{M}(\bar{m})$ with unit $n$ and $K - |\bar{M}| - 1$ other units in state $0$ has probability of $1$. Therefore:
$$
\bar{\pi}(|\bar{M}|) \geq \frac{\sum_{b = K - |\bar{M}| - 1}^{N - |\bar{M}| - 1} \binom{N - |\bar{M}| - 1}{b}}{2^{N - |\bar{M}|}} = \frac{\sum_{b = K - |\bar{M}|}^{N - |\bar{M}|} \binom{N - |\bar{M}| - 1}{b-1}}{2^{N - |\bar{M}|}}.
$$
To find the MIP, we need to find the minimum normalized difference between partitioned and unpartitioned repertoires for the pair $Z$ and $M$:
\begin{adjustwidth}{-2.25in}{0cm}
\begin{equation}
\label{eq:parition_sum_intra}
\frac{1}{\mathcal{N}(\theta)} 
\pi_e(Z = z \mid m ) 
\pospart{ \log_2(\frac{\pi_e(Z = z \mid m )}{\pi_e^\theta(Z = z \mid m )}) }
= \frac{1}{\mathcal{N}(\theta)} \log_2(\frac{1}{\prod_i \pi_e^\theta(Z_i = z_i \mid m)}) 
= -\frac{1}{\mathcal{N}(\theta)} \sum_{Z_i \in Z} \log_2(\pi_e(Z_i = z_i \mid m_i)) 
\end{equation}
\end{adjustwidth}

$m_i$ is the mechanism connected to $Z_i$ after the partition and $\pi_e^\theta(Z_i = z_i \mid m_i)$ is either $\pi(|M_i|)$ and $\bar{\pi}(|M_i|))$.
The lower bounds for both $\pi(|\bar{M}|)$ and $\bar{\pi}(|\bar{M}|))$, therefore the upper bound for the sum, are achieved when:
\begin{equation}
    \label{eq:construction_0}
    p(Z_n=0 \mid s) = 0,  \forall s \notin \bigcup_j \mathcal{M}(m_{n,j}), n = 1,\dots, N.
\end{equation}
Eq \eqref{eq:construction} and Eq \eqref{eq:construction_0} provide us with $p(Z_n=0 \mid s)$ for all system states $s$ and purview units $n$, which gives us the full TPM.

We can even further decompose the sum in Eq \eqref{eq:parition_sum_intra} into sum over individual connections. Assume we are removing connections from the mechanism $m$ to the purview unit $z_i$ one by one until we arrive at the mechanism after the partitioning $m_i$. Let us denote the intermediate steps as $m_i^{(0)}, m_i^{(1)}, m_i^{(2)}, \dots, m_i^{(N_i)}$, where $m_i^{(0)} = m$ and $m_i^{(N_i)} = m_i$ and $N_i$ is the number of connections cut from $z_i$. We can write the normalized partitioned informativeness as:
\begin{equation}
\small
\label{eq:parition_sum_connection}
\frac{1}{\mathcal{N}(\theta)} 
\pospart{ \sum_{Z_i \in Z} \sum_{c = 1}^{N_i} \log_2(\frac{\pi_e(Z = z \mid m _i^{(c-1)})}{\pi_e(Z = z \mid m _i^{(c)})}) } 
= \frac{1}{\mathcal{N}(\theta)} 
 \sum_{Z_i \in Z} \sum_{c = 1}^{N_i} \log_2(\frac{\pi_e(Z = z \mid m _i^{(c-1)})}{\pi_e(Z = z \mid m _i^{(c)})}) 
\end{equation}
Positive part operator can be removed as both $\pi(|\bar{M}|)$ and $\bar{\pi}(|\bar{M}|))$ decrease as the size of the mechanism $|\bar{M}|$ increases, making all the terms inside the sum positive. Eq \eqref{eq:parition_sum_connection} rewrites the normalized partitioned informativeness as the average informativeness gain over the individual connections severed  by the partition. 

From Lemma \ref{lem:pi_ratio}, we know the information gain per connection can at most be $1$.  Now, using $\pi(|\bar{M}|)$ and $\bar{\pi}(|\bar{M}|))$, we can make a few observations about the information gain of removing different connections in the system under consideration. First, we show cutting a self connection can achieve the maximum information gain of $1$. 
This means removing the self connection cannot decrease the normalized partitioned informativeness, as adding $1$ to a set of numbers that can at most be $1$ cannot decrease the average. 
To calculate the information gain for cutting a self connection,  we should compare $\pi(|M_i|)$ to $\bar{\pi}(|M_i| - 1)$, \emph{i.e.,} cutting one input connection from unit $Z_i$ such that the mechanism connected to it no longer contains it:
$$
\begin{aligned}
\pi(|M_i|))
&= \frac{\sum_{b = K - |M_i|}^{N - |M_i|} \binom{N - |M_i|}{b}}{2^{N - |M_i|}}
= \frac{2\sum_{b = K - |M_i|}^{N - |M_i|} \binom{N - |M_i|}{b}}{2^{N - |M_i| + 1}}\\
&= 2\frac{\sum_{b = K - (|M_i|-1)-1}^{N - (|M_i|-1)-1} \binom{N - (|M_i|-1)-1}{b}}{2^{N - (|M_i|-1)}} = 2 \bar{\pi}(|M_i| - 1)).
\end{aligned}
$$

This provides us with the proof for Theorem \ref{thm:order_k_fully_det_effect}.
\orderkdet*
\begin{proof}
Achieving the maximum value is only possible if the MIP is the complete partition. Since, as shown above, not having the self connections in the cut for $1 < K < N$ gives us a smaller normalized partitioned informativeness, the complete partition cannot be the MIP.
\end{proof}

We can similarly calculate the informativeness gain for a lateral connection as
$
\log_2(\frac{\pi(|\bar{M}|)}{\pi(|\bar{M}|-1)})
$
or
$
\log_2(\frac{\bar{\pi}(|\bar{M}|)}{\bar{\pi}(|\bar{M}|-1)})
$, depending on if the self connection is intact or not, respectively. In the first case we have:
\begin{adjustwidth}{-2.25in}{0cm}
$$
\small
\log_2(\frac{\pi(|\bar{M}|)}{\pi(|\bar{M}|-1)})
= \log_2(
2
\frac{
\sum_{b = K - |\bar{M}|}^{N - |\bar{M}|} \binom{N - |\bar{M}|}{b})}
{\sum_{b = K - |\bar{M}| + 1}^{N - |\bar{M}| + 1} \binom{N - |\bar{M}| + 1}{b}}
)
= 1 + \log_2(
\frac{
\sum_{b = 0}^{N - K} \binom{N - |\bar{M}|}{b})}
{\sum_{b = 0}^{N - K} \binom{N - |\bar{M}| + 1}{b}}
)
= 1 - \log_2(
\frac
{\sum_{b = 0}^{N - K} \binom{N - |\bar{M}| + 1}{b}}
{\sum_{b = 0}^{N - K} \binom{N - |\bar{M}|}{b})}
).
$$
\end{adjustwidth}
This shows the information gain of $1$ is only achievable for the special case of $K = N$, which is the trivial case of no intra-order trade-off. Furthermore, since the ratio inside $\log$ is decreasing with $N - |\bar{M}|$, the information gain increases as we cut more lateral connections. Same results hold for $
\log_2(\frac{\bar{\pi}(|\bar{M}|)}{\bar{\pi}(|\bar{M}|-1)})
$ as well.

This helps us to narrow down the scope of our search for the MIP. Consider the cut where one mechanism unit is removed from all the $K$ purview units. This cut removes the minimum number of connections among all the valid cuts. It removes $K$ connections, $1$ self connection and $K-1$ lateral connections. Any other valid cut that removes at least one self connection can only increase the average, compared to this cut. This is because, as shown above, cutting more self connections can only increase the average, except the special case of $K=N$, and cutting more lateral connections from a purview unit also only increases the average. Therefore, the special cut discussed above has smaller average information gain, compared to any other cut that removes at least one self connection. This narrows our search to only this cut and the cuts that do not remove self connections. 

We can further show that among the partitions that do not cut self connections bi-partitions have the smallest average information gain. To see this, starting from any partition with more than two parts, merging the two smallest parts decreases the average. This is because the purview units in the smallest parts have the maximum number of connections removed from them. As already shown, the information gain increases as we cut more lateral connections, therefore by merging the smallest parts we can avoid removing the connections with maximum information gain, which can only decrease the average. This narrows the candidate partitions to $\frac{K}{2} + 1$ partitions, \emph{i.e.,} linear growth with size, and makes it computationally feasible to evaluate for bigger networks. 

\subsection[\appendixname~\thesection]{Relations' integrated information}
\label{app:relation-proofs}
In this section, we provide the solution to the optimization problem in \eqref{eq:optimization}:
\begin{equation}
\tag{\ref{eq:optimization} repeated}
\begin{aligned}
\max
&\sum_{i = 1}^{|\mathcal{Z}(o)|} \frac{\varphi_{(i)}}{|z_{(i)}|} (2^{|\mathcal{Z}(o)| - i} - 1), \\
\text{subject to}
& \sum_{i = 1}^{|\mathcal{Z}(o)|} \frac{\varphi_{(i)}}{|z_{(i)}|} \leq S(o)
\end{aligned}
\end{equation}
$\frac{\varphi_{(i)}}{|z_{(i)}|}$ is sorted indexing of the elements in $\mathcal{Z}(o)$, such that $(z_{(1)}, \varphi_{(1)})$ has the smallest $\frac{\varphi}{|z|}$ ratio, $(z_{(2)}, \varphi_{(2)})$ has the second smallest $\frac{\varphi}{|z|}$ ratio, and so on.
Thus, we need to formalize the constraint that $\frac{\varphi_{(i)}}{|z_{(i)}|}$ is non-decreasing and non-negative by a change of variable as follows:
$$
\begin{aligned}
\frac{\varphi_{(1)}}{|z_{(1)}|} = x_1,
\frac{\varphi_{(2)}}{|z_{(2)}|} = x_1 + x_2,
\dots,
\frac{\varphi_{(i)}}{|z_{(i)}|} = \sum_{j=1}^i x_j, && x_j \geq 0, \forall j
\end{aligned}
$$
In other words, we are defining $\frac{\varphi_{(i)}}{|z_{(i)}|} = \frac{\varphi_{(i-1)}}{|z_{(i-1)}|} + x_i$, where $x_i$ is a non-negative value. This translates the problem in \eqref{eq:optimization} into:
\begin{equation}
\begin{aligned}
\max
&\sum_{i = 1}^{|\mathcal{Z}(o)|} (2^{|\mathcal{Z}(o)| - i} - 1) \sum_{j=1}^i x_j, \\
\text{subject to}
& \sum_{i = 1}^{|\mathcal{Z}(o)|} \sum_{j=1}^i x_j 
= \sum_{i = 1}^{|\mathcal{Z}(o)|} (|\mathcal{Z}(o)| - i + 1) x_i \leq S(o) \\
& x_i \geq 0,\ i=1,\dots, |\mathcal{Z}(o)|.
\end{aligned}
\end{equation}
This is a relatively easy problem to solve, as it includes maximizing a linear function of $|\mathcal{Z}(o)|$ variables, given $|\mathcal{Z}(o)| +1$ linear constraints. The solution to the problem is one of the $|\mathcal{Z}(o)| +1$ vertices defined by the constraints. The first vertex is $x_i = 0, \forall i$, which is the trivial minimum of the problem. The rest of the vertices have the same property that for some $k$, $x_k = \frac{S(o)}{|\mathcal{Z}(o)| - k + 1}$ and $x_i = 0, i \neq k$. For such vertex, the value of the objective function is:
$$
\frac{S(o)}{|\mathcal{Z}(o)| - k + 1} \sum_{i = k}^{|\mathcal{Z}(o)|} (2^{|\mathcal{Z}(o)| - i} - 1) 
= 
S(o) 
\left(
\frac{2^{|\mathcal{Z}(o)| - k + 1}}{|\mathcal{Z}(o)| - k + 1}
(1 - \frac{1}{2^{|\mathcal{Z}(o)|}}) - 1
\right).
$$
This value is monotonically decreasing with $k$, which means the maximum occurs at the vertex where $x_1 = \frac{S(o)}{|\mathcal{Z}(o)|}$ and $x_i = 0$ for $i=2,\dots,|\mathcal{Z}(o)|$. This corresponds to a system where all the distinctions in $\mathcal{Z}(o)$ have the same $\frac{\varphi}{|z|}$ value, i.e. $\frac{\varphi_{(i)}}{|z_{(i)}|} = \frac{S(o)}{|\mathcal{Z}(o)|}, \forall i$. In other words, if we distribute the sum, $S(o)$, such that the distinction integrated information is proportional to its $|z^*_c \cup z^*_e|$ for all the distinctions, we can achieve the maximum value of:
$$
S(o) 
\left(
\frac{2^{|\mathcal{Z}(o)|}}{|\mathcal{Z}(o)|}
(1 - \frac{1}{2^{|\mathcal{Z}(o)|}}) - 1
\right).
$$

\end{document}